\documentclass[draftcls,onecolumn, 12pt]{IEEEtran}

\usepackage{amsmath,amssymb,amsfonts,mathrsfs}
\usepackage{cite}
\usepackage{array}
\usepackage{enumerate}
\usepackage[dvips]{graphicx}
\usepackage{pstricks}
\usepackage{psfrag}
\usepackage[dvips,colorlinks,linkcolor=black]{hyperref}

\newtheorem{Corollary}{Corollary}
\newtheorem{Proposition}{Proposition}
\newtheorem{Lemma}{Lemma}
\newtheorem{Definition}{Definition}
\newtheorem{Theorem}{Theorem}
\newtheorem{Example}{Example}

\newtheorem{Assumption}{Assumption}

\newcommand{\EndExample}{{$\square$}}

\newcommand{\Real}{\mathbb{R}}

\newcommand{\ud}{\mathrm{d}}
\newcommand{\E}{\mathbb{E}}

\newcommand{\Q}{\mathbb{Q}}
\renewcommand{\P}{\mathbb{P}}

\newcommand{\ofrac}[1]{{\frac{1}{#1}}}
\newcommand{\ddfrac}[2]{{\frac{\ud #1}{\ud #2}}}

\newcommand{\tn}[1]{{^{(#1)}}}
\newcommand{\KLD}[2]{{\mathrm{D}({#1}\, \| \, {#2})}}

\newcommand{\Pin}[1]{{\P_{1,n}^{(#1)}}}
\newcommand{\Pon}[1]{{\P_{0,n}^{(#1)}}}
\newcommand{\Lambdahs}[2]{{\Lambda_{#1,#2}^*(\gamma, t^{(#2)})}}

\begin{document}

\title{Data Fusion Trees for Detection: Does Architecture Matter?
    \thanks{
    This research was supported, in part, by the
    National Science Foundation under contracts ECS-0426453 and
    ANI-0335256, the Charles Stark Draper Laboratory Robust Distributed
    Sensor Networks Program, and an Office of Naval Research Young
    Investigator Award N00014-03-1-0489.
    A preliminary version of this paper was presented at the 44th Annual Allerton Conference on Communication,
    Control, and Computing, Monticello, Illinois, September 2006.
    W.P.\ Tay, J.N.\ Tsitsiklis and M.Z.\ Win are with the
    Laboratory for Information and Decision Systems, MIT, Cambridge, MA, USA.
    E-mail: \texttt{\{wptay, jnt, moewin\}@mit.edu}
    }
}
\author{Wee~Peng~Tay,~\IEEEmembership{Student Member,~IEEE,}
        John~N.~Tsitsiklis,~\IEEEmembership{Fellow,~IEEE,}
        \and and~Moe~Z.~Win,~\IEEEmembership{Fellow,~IEEE}
}

\markboth
   {Submitted to IEEE Trans. Information Theory}
   {Tay \MakeLowercase{\textit{et al.}}: Data Fusion Trees for Detection: Does Architecture Matter?}

\maketitle \thispagestyle{empty}


\begin{abstract}
We consider the problem of decentralized detection in a network consisting of a large number of
nodes arranged as a tree of bounded height,
under the assumption of conditionally independent, identically distributed observations.
We characterize the optimal
error exponent under a Neyman-Pearson formulation.
We show that the Type II error probability decays exponentially fast with the number of nodes, and the optimal error exponent is often the same as that corresponding to a parallel configuration.
We provide sufficient, as well as necessary, conditions for this to happen.
For those networks satisfying the sufficient conditions,
we propose a simple strategy that nearly achieves the optimal error exponent,
and in which all non-leaf nodes need only send 1-bit messages.
\end{abstract}


\section{Introduction}\label{sect:Introduction}

Most of the decentralized detection literature has been concerned with characterizing optimal detection strategies for particular sensor configurations; the comparison of
the detection performance of different configurations is a rather unexplored area. We bridge this gap by
considering the asymptotic performance of bounded height tree networks.
We analyze the dependence of the optimal error exponent on the
network architecture, and characterize the optimal error exponent for a large class
of tree networks.

The problem of optimal
decentralized detection has attracted a lot of interest over the last twenty-five years.
Tenney and Sandell~\cite{TenSan:81} are the first to consider a decentralized detection system in which each of several
sensors makes an observation and sends a summary
(e.g., using a quantizer or other ``transmission function'') to a fusion center. Such a system is to be contrasted to a {\em centralized} one, where the raw observations are transmitted directly to the fusion center. The framework  introduced in \cite{TenSan:81} involves a ``star topology'' or ``parallel configuration'': the fusion center is regarded as the root of a tree, while the sensors
are the leaves, directly connected to the root. Several pieces of work follow, e.g., \cite{ChaVar:86,PolTsi:90,WilWar:92,Tsi:93a,Tsi:93,IrvTsi:94,VisVar:97,CheVar:02,CheWil:05,Kas:06,LiuChe:06}, all of which study the parallel configuration under a Neyman-Pearson or Bayesian criterion. A common goal of these
references is to characterize the optimal transmission function, where optimality usually refers to the minimization of the probability of error or some other cost function at the fusion center. A typical result is that under the assumption of (conditionally) independent sensor observations, likelihood ratio quantizers are optimal; see \cite{Tsi:93} for a summary of such results.

The study of sensor networks other than the parallel configuration is initiated in~\cite{EkcTen:82}, which considers a tandem configuration, as well as more general tree configurations, and characterizes optimal
transmission strategies under a Bayesian formulation. Tree configurations are also discussed in \cite{VisThoTum:88,ReiNol:90a,TanPatKle:91,TanPatKle:93,PapAth:92a,PetPatKle:94,AlhVar:95,LinCheVar:05}, under various performance objectives.
In all but the simplest cases, the exact form of optimal strategies in tree configurations is difficult to derive. Most of these references focus on  person-by-person (PBP) optimality and obtain necessary, but not sufficient, conditions for an optimal strategy. When the transmission functions are assumed to be finite-alphabet quantizers, typical results establish that under
a conditional independence assumption,
likelihood ratio quantizers are PBP optimal. However, finding the optimal quantizer thresholds requires the solution of a nonlinear system of equations, with as many equations as there are thresholds. As a consequence, computing the optimal thresholds or characterizing the overall performance is hard, even for networks of moderate size.

Because of these difficulties, the analysis and comparison of large sensor networks
is apparently tractable only in an asymptotic regime that focuses
on the rate of decay of the error probabilities as the number of sensors increases.
For example, in the Neyman-Pearson framework, one can focus on minimizing the error exponent \footnote{Throughout this paper, $\log$ stands for the natural logarithm.}
\begin{align*}
g = \limsup_{n\to\infty} \ofrac{n} \log \beta_n,
\end{align*}
where $\beta_n$ is the Type II error probability at the fusion center and $n$ is the number of sensors,
while keeping the Type I error probability less than some given threshold.
Note our convention that error exponents are negative numbers.
The magnitude of the error exponent, $|g|$, is commonly referred to as the rate of decay of the Type II error probability.
A larger $|g|$ would translate to a faster decay rate, hence a better detection performance.
This problem has been studied in~\cite{Tsi:88}, for the case of a parallel configuration with a large number of sensors that receive
independent, identically distributed (i.i.d.) observations.

The asymptotic performance of another special configuration, involving $n$ sensors arranged in tandem, has been studied in
\cite{HelCov:70, Cov:69, PapAth:92}, under a Bayesian formulation.
Necessary and sufficient conditions for the error probability to decrease to zero as $n$ increases have been derived. However, even when the error probability decreases to zero, it apparently does so at a sub-exponential rate (see \cite{TayTsiWin:C07c} for such a result for the Bayesian case).
Accordingly, \cite{PapAth:92} argues that the tandem configuration is inefficient and suggests that as the number of sensors increases, the network ``should expand more in a parallel than in [a] tandem'' fashion.

Even though the error probabilities in a parallel configuration decrease exponentially, the energy consumption of having each sensor transmit directly to the fusion center can be too high. The energy consumption can be reduced
by setting up a directed spanning in-tree, rooted at the fusion center.  In a tree configuration, each non-leaf node combines its own observation (if any) with the messages it has received and forms a new message, which it transmits to another node. In this way, information from each node is propagated along a multi-hop path to the fusion center, but the information is ``degraded'' along the way.
For the case where observations are obtained only at the leaves, it is not hard to see that the detection performance of such a tree cannot be better than that of a parallel configuration with the same number of leaves.

In this paper, we investigate the detection performance of a tree configuration under a Neyman-Pearson criterion.
We restrict to trees with bounded height for two reasons. First, without a restriction on the height of the tree, performance can be poor (this is
exemplified by tandem networks in which, as remarked above, the error probability seems to decay at a sub-exponential rate). Second, bounded height translates to a bound on
the delay until information reaches the fusion center.

As it is not apparent that the Type II error probability decays exponentially fast with the
number of nodes in the network, we first show that under the bounded height assumption, exponential decay
is possible.
We then obtain the rather
counterintuitive result that
if leaves dominate (in the sense that asymptotically almost all nodes are leaves), then
bounded height trees have the same asymptotic performance as the parallel configuration,
even in non-trivial cases. (Such an equality is clear in some trivial cases, e.g., the configuration shown in
Figure \ref{fig:TrivialCase}, but is unexpected in general.) This result has important ramifications: a system designer can reduce the energy consumption in a network (e.g., by employing an $h$-hop spanning tree that minimizes the overall energy consumption), without losing detection efficiency, under certain conditions.

\begin{figure}[!htb]
\begin{center}
\includegraphics[scale=1]{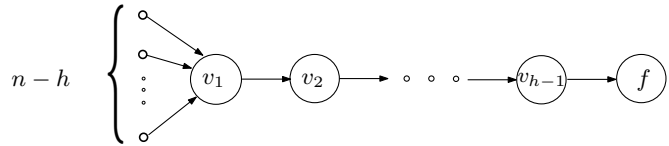}
\caption{A tree network of height $h$, with $n-h$ leaves. Its error probability is no larger than that of a parallel configuration with $n-h$ leaves  and a fusion center. If $h$ is bounded while $n$ increases, the optimal error exponent
is the same as for a parallel configuration with $n$ leaves.
}\label{fig:TrivialCase}
\end{center}
\end{figure}

We also provide a strategy in which each non-leaf node sends only a 1-bit message, and which nearly
achieves the same performance as the parallel configuration.
These results are counterintuitive for the following reasons: 1) messages are compressed to only one bit at each
non-leaf node so that
``information'' is lost along the way, whereas in the parallel configuration, no such compression occurs;
2) even though leaves dominate, there is no reason why the error exponent will be determined
solely by the leaves. For example, our discussion in Section \ref{subsect:DiscussionSufficient} indicates that without
the bounded height assumption, or
if a Bayesian framework is assumed instead of the Neyman-Pearson formulation, then a generic tree network (of height greater than 1)
performs strictly worse than a parallel configuration, even if leaves dominate.

Finally, under a mild additional assumption on the allowed transmission functions, we find that the sufficient conditions
for achieving the same error exponent as a parallel configuration, are also necessary.

The rest of this paper is organized as follows. In Section \ref{sect:ProblemFormulation}, we present our model in
detail. In Section \ref{sect:NeymanPearson}, we state the Neyman-Pearson problem, provide some motivating examples, and
state the main results.
In Section \ref{sect:SimpleRelay}, we consider  ``relay trees,'' in which observations are only made at the leaves.
In Section \ref{sect:OptimalError}, we prove the main results.
Finally, in Section \ref{sect:Conclusion}, we summarize and offer some concluding
remarks.

\section{Problem Formulation}\label{sect:ProblemFormulation}

In this section, we introduce the model 
and the required notation.
We consider a decentralized binary detection problem involving $n-1$ sensors and a fusion center; we will be interested in the case where $n$ increases to infinity.
We are given two probability spaces $(\Omega, \mathcal{F}, \P_0)$ and
$(\Omega, \mathcal{F}, \P_1)$, associated with two hypotheses $H_0$ and $H_1$.
We use $\E_j$ to denote the expectation operator with respect to $\P_j$.
Each sensor $v$ observes a random variable $X_v$ taking values in some set $\mathcal{X}$.
Under either hypothesis $H_j$, $j=0,1$, the random variables $X_v$ are i.i.d., with marginal distribution $\P_j^X$.

\subsection{Tree Networks} \label{subsect:TreeNetwork}

The configuration of the sensor network is represented by a directed tree
$T_n = (V_n, E_n)$. Here, $V_n$ is the set of nodes, of cardinality $n$,
and $E_n$ is the set of directed arcs of the tree.
One of the nodes (the ``root'') represents the fusion center, and the remaining $n-1$ nodes represent the remaining sensors.
We will always use the special symbol $f$ to denote the root of $T_n$.
We assume that the arcs
are oriented so that they all point towards the fusion center.
In the sequel, whenever we use the term ``tree'', we mean a directed, rooted tree as described above.

We will use the terminology ``sensor'' and ``node'' interchangeably. Moreover, the fusion
center $f$ will also be called a sensor, even though it plays the special role
of fusing; whether the fusion center makes its own observation or not is irrelevant, since we are working in the
large $n$ regime, and we will assume it does not.

We say that node $u$ is a \emph{predecessor} of
node $v$ if there exists a directed path from $u$ to $v$. In this case, we also say that $v$ is a
\emph{successor} of $u$. An \emph{immediate predecessor} of node $v$ is a node $u$ such that $(u,v) \in E_n$. An immediate successor is similarly defined.
Let the set of immediate predecessors of $v$ be $C_n(v)$.
If $v$ is a leaf, $C_n(v)$ is naturally defined to be empty.
The \emph{length} of a path is defined as the number of arcs in the path. The \emph{height} of the tree $T_n$ is the length of the longest path from a leaf to the root, and will be denoted by $h_n$.

Since we are interested in asymptotically large values of $n$, we will consider
a \emph{sequence} of trees $(T_n)_{n \geq 1}$. While we could think of the sequence as representing the evolution of the network as sensors are added, we do not require the sequence $E_n$ to be an increasing sequence of sets; thus, the addition of a new sensor to $T_n$ may result in some edges being deleted and some new edges being added.
We define the height of a sequence of trees to be $h = \sup_{n\geq1} h_n$.
We are interested in tree sequences of bounded height, i.e., $h < \infty$.

\begin{Definition}[$h$-uniform tree]\label{def:HeightUniform}
A tree $T_n$ is said to be $h$-uniform if the length of every path from a leaf to the root is
exactly $h$. A sequence of trees $(T_n)_{n\geq1}$ is said to be $h$-uniform if there exists some $n_0 <\infty$,
so that for all $n \geq n_0$, $T_n$ is $h$-uniform.
\end{Definition}

For a tree with height $h$, we say that a node is at \emph{level} $k$ if it is connected to the fusion center via a path of length $h-k$. Hence the fusion center $f$ is at level $h$, while in an $h$-uniform tree, all leaves are at level 0.

Let $l_n(v)$ be the number of leaves of the sub-tree rooted at the node $v$.
(These are the leaves whose path to $f$ goes through $v$.)
Thus, $l_n(f)$ is the total number of leaves.
Let $p_n(v)$ be the total number of predecessors of $v$, i.e., the total number of nodes in the sub-tree
rooted at $v$, not counting $v$ itself. Thus, $p_n(f)=n-1$.
We let $A_n \subset V_n$ be the set of nodes whose immediate predecessors include leaves of the
tree $T_n$. Finally, we let $B_n \subset A_n$ be the set of nodes all of whose predecessors are leaves; see Figure \ref{fig:AB}.

\begin{figure}[!htb]
\begin{center}
\includegraphics[scale=1]{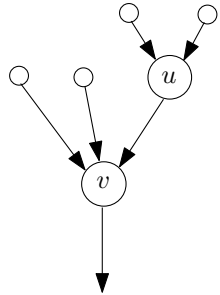}
\caption{Both nodes $v$ and $u$ belong to the set $A_n$, but only node $u$ belongs to the set $B_n$.}\label{fig:AB}
\end{center}
\end{figure}

\subsection{Strategies} \label{subsect:AdmissibleStrategy}

Given a tree $T_n$, consider a node $v\neq f$.
Node $v$ receives messages $Y_{u}$ from every  $u \in C_n(v)$ (i.e., from its immediate predecessors).
Node $v$ then uses a transmission function $\gamma_v$ to encode and transmit a summary
$Y_{v} = \gamma_v(X_v, \{Y_{u}: u \in C_n(v)\})$ of its own observation $X_v$, and of the received messages $\{Y_{u}: u \in C_n(v)\}$, to its immediate successor.\footnote{To simplify the notation, we suppress the dependence of $X_v$, $Y_v$, $\gamma_v$, etc.\ on $n$.}
We constrain all messages to be symbols in a fixed alphabet $\mathcal{T}$. Thus, if the in-degree of $v$ is $|C_n(v)|=d$, then the transmission function $\gamma_v$ maps $\mathcal{X} \times \mathcal{T}^d$ to $\mathcal{T}$.
Let $\Gamma(d)$ be a given set of transmission functions that the node $v$ can choose from.
In general, $\Gamma(d)$ is a subset of the set of all possible mappings from $\mathcal{X} \times \mathcal{T}^d$ to $\mathcal{T}$.
For example, $\Gamma(d)$ is often assumed to be the set of quantizers whose outputs are the result of comparing
likelihood ratios to some
thresholds (cf.\ the definition of a Log-Likelihood Ratio Quantizer in Section \ref{subsect:Assumptions}).
For convenience, we
denote the set of transmission functions for the leaves, $\Gamma(0)$, by $\Gamma$.
We assume that all transmissions are perfectly reliable.

Consider now the root $f$, and suppose that it has $d$ immediate predecessors. It receives messages from its immediate predecessors, and based on this information, it decides between the two hypotheses $H_0$ and $H_1$,
using a fusion rule $\gamma_f: \mathcal{T}^d \mapsto \{0,1\}$.\footnote{
Recall that in centralized Neyman-Pearson detection, randomization can reduce the Type II error probability.
Therefore, in general, the fusion center uses a randomized fusion rule to make its decision.
Similarly, the transmission functions $\gamma_v$ used by each node $v$, can also be
randomized. We avoid any discussion of randomization to simplify the exposition, and because randomization is not required
asymptotically, as will become apparent in Section \ref{sect:OptimalError}.}
Let $Y_{f}$ be a binary-valued random variable indicating the decision of the fusion center.

We define a {\em strategy} for a tree $T_n$, with $n-1$ nodes and a fusion center, as
a collection of transmission functions, one for each node, and a fusion rule.
In some cases, we will be considering strategies in which only the leaves make observations; every other node $v$ simply fuses the messages it has received, and forwards a message
$Y_{v} = \gamma_v( \{Y_{u}: u \in C_n(v)\})$ to its immediate successor.
A strategy of this type will be called a {\it relay strategy}.
A tree network in which we restrict to relay strategies will be called a {\it relay tree}.
If in addition, the alphabet $\mathcal{T}$ is binary, we will
use the terms \emph{1-bit relay strategy} and \emph{1-bit relay tree}.
Finally, in a relay tree, nodes other than the root and the leaves will be called \emph{relay nodes}.

\section{The Neyman-Pearson Problem}\label{sect:NeymanPearson}

In this section, we formulate the Neyman-Pearson decentralized detection problem in a tree network. We provide some motivating examples, and introduce our assumptions. Then, we give a summary of the main results.

Given a tree $T_n$,
we require that the
Type I error probability $\P_0(Y_{f}=1)$ be no more than a given $\alpha \in (0,1)$.
A strategy is said to be {\em admissible} if it meets this constraint.
We are interested in minimizing the
Type II error probability $\P_1(Y_{f}=0)$.
Accordingly, we define
$\beta^*(T_n)$ as the infimum of  $\P_1(Y_{f}=0)$, over all admissible strategies.
Similarly, we define $\beta_R^*(T_n)$ as the infimum of $\P_1(Y_{f}=0)$, over all admissible relay strategies.
Typically,
$\beta^*(T_n)$ or $\beta^*_R(T_n)$ will converge to zero as $n \to \infty$. We are interested in the question of whether such convergence takes place exponentially fast, and in the exact value of the Type II error exponent, defined by
$$
g^*=\limsup_{n\to\infty} \frac{1}{n}\log \beta^*(T_n),
\qquad g^*_R =\limsup_{n\to\infty} \frac{1}{l_n(f)}\log \beta^*_R(T_n).
$$
Note that in the relay case, we use the total number of leaves $l_n(f)$ instead of $n$ in the definition
of $g^*_R$. This is because only the leaves make observations and therefore, $g^*_R$ measures
the rate of error decay per observation.

We denote the Kullback-Leibler (KL) divergence of two probability measures, $\P$ and $\Q$, as
\begin{align*}
\KLD{\P}{\Q} = \E^{\P}\Big[\log \ddfrac{\P}{\Q}\Big],
\end{align*}
where $\E^\P$ is the expectation operator with respect to (w.r.t.) $\P$.
Suppose that $X$ is a sensor observation.
For any $\gamma\in\Gamma$, let the distribution of $\gamma(X)$ be $\P_j^\gamma$.
Note that $-\KLD{\P_0^\gamma}{\P_1^\gamma} \leq 0 \leq \KLD{\P_1^\gamma}{\P_0^\gamma}$, with both inequalities being strict as long as the measures $\P_0^\gamma$ and $\P_1^\gamma$ are not indistinguishable.

In the classical case of a parallel configuration, with $n-1$ leaves directly connected to the fusion center, the optimal error exponent, denoted as $g_P^*$, is given by \cite{Tsi:88}
\begin{align} \label{errorexpclassical}
g_P^*=\lim_{n \to \infty} \ofrac{n} \log \beta^*(T_n)
= - \sup_{\gamma \in \Gamma} \KLD{\P_0^\gamma}{\P_1^\gamma},
\end{align}
under Assumptions \ref{assumpt:EquivalentMeasures}-\ref{assumpt:BoundedDivergence}, stated in Section \ref{subsect:Assumptions} below.

Our objective is to study $g^*$ and $g^*_R$ for different sequences of trees. In particular, we wish to obtain bounds on these quantities, develop conditions under which they are strictly negative (indicating exponential decay of error probabilities), and develop conditions under which they are equal to $g^*_P$. At this point, under Assumptions \ref{assumpt:EquivalentMeasures}-\ref{assumpt:BoundedDivergence}, we can record two relations that are always true:
\begin{equation}
g^*_P \leq g^*_R, \qquad -\KLD{\P_0^X}{\P_1^X}\leq g^* \leq z g^*_R, \label{eq:gg}
\end{equation}
where $z = \liminf\limits_{n\to\infty} l_n(f)/n$.
The first inequality is true because all of the combining of messages that takes place in a relay network can be carried out internally, at the fusion center of a parallel network with the same number of leaves.
The inequality $-\KLD{\P_0^X}{\P_1^X} \leq g^*$ follows from the fact that
$-\KLD{\P_0^X}{\P_1^X}$ is the classical error exponent in a centralized system where all raw observations are transmitted directly to the the fusion center. Finally, the inequality $g^* \leq z g^*_R$ follows because an optimal strategy is at least as good as an optimal relay strategy; the factor of $z$
arises because we have normalized $g^*_R$ by $l_n(f)$ instead of $n$.

For a sequence of trees of the form shown in Figure~\ref{fig:TrivialCase}, it is easily seen that $g^*=g^*_R=g^*_P$.
In order to develop some insights into the problem, we now consider some less trivial examples.

\subsection{Motivating Examples}\label{subsect:Motiv}

In the following examples, we restrict to relay strategies for simplicity, i.e., we are interested in characterizing the error exponent $g_R^*$.
However, most of our subsequent results hold without such a restriction, and similar statements
can be made about the error exponent $g^*$ (cf.\ Theorem \ref{thm:MainResults}).

\begin{Example}\label{eg:TwoHeight}
Consider a 2-uniform sequence of trees, as shown in Figure \ref{fig:TwoHeight}, where each node $v_i$ receives messages
from $m=(n-3)/2$ leaves (for simplicity, we assume that $n$ is odd).
\begin{figure}[!htb]
\begin{center}
\includegraphics[scale=1]{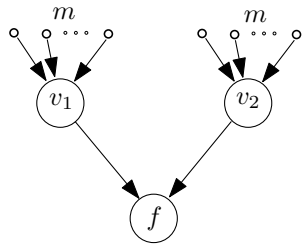}
\caption{A 2-uniform tree with two relay nodes.}\label{fig:TwoHeight}
\end{center}
\end{figure}

Let us restrict to 1-bit relay strategies.
Consider the fusion rule that declares $H_0$ iff both $v_1$ and $v_2$ send a $0$. In order to keep the  Type I error
probability bounded by $\alpha$, we view the message by each $v_i$ as a local decision about the hypothesis,
and require that its local Type I
error probability be bounded by $\alpha/2$. Furthermore, by viewing the sub-tree rooted at $v_i$ as a parallel configuration,
we can design strategies for each sub-tree
so that
\begin{align}
\lim_{n\to\infty} \ofrac{m}\log \P_1( Y_{v_i}=0) = g^*_P. \label{LocalErrExp}
\end{align}
At the fusion center, the Type II error exponent is then given by
\begin{align*}
\lim_{n\to\infty} \ofrac{n}\log \beta_n
& = \lim_{n\to\infty} \ofrac{n} \log \P_1( Y_{v_1}=0, Y_{v_2}=0) \\
& = \ofrac{2}\lim_{n\to\infty} \ofrac{m} \log \P_1( Y_{v_1}=0)
+ \ofrac{2}\lim_{n\to\infty} \ofrac{m} \log \P_1(Y_{v_2}=0) \\
& = g^*_P,
\end{align*}
where the last equality follows from (\ref{LocalErrExp}). This shows that the
Type II error probability falls exponentially and,
more surprisingly, that $g^*_R\leq g^*_P$. In view of Eq.\ (\ref{eq:gg}), we have $g^*_R=g^*_P$.
It is not difficult to generalize this conclusion to all sequences of trees in which the number $n-l_n(f)-1$ of relay nodes is bounded. For such sequences, we will also see that $g^*=g^*_R$ (cf.\ Theorem \ref{thm:MainResults}(\ref{MR:4})).\hfill \EndExample
\end{Example}

\begin{Example}\label{eg:InfiniteRelay}
We now consider an example in which the number of relay nodes grows with $n$. In Figure \ref{fig:InfiniteRelay}, we let
both $m$ and $N$ be increasing functions of $n$ (the total number of nodes), in a manner to be made explicit shortly.

\begin{figure}[!htb]
\begin{center}
\includegraphics[scale=1]{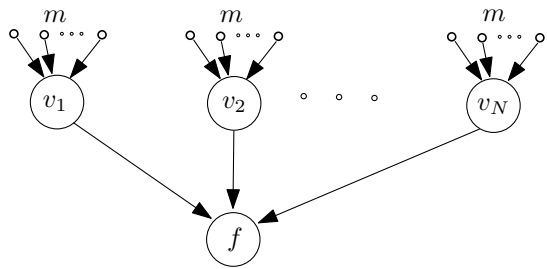}
\caption{A 2-uniform tree with a large number of relay nodes. }\label{fig:InfiniteRelay}
\end{center}
\end{figure}

Let us try
to apply a similar argument as in Example \ref{eg:TwoHeight}, to see whether the optimal exponent of the parallel configuration can be achieved with a relay strategy, i.e., whether $g^*_R=g^*_P$.
We let each node $v_i$ use a local Neyman-Pearson
test. We also let the fusion center declare $H_0$ iff it
receives a 0 from all relay sensors.
In order to have a hope of achieving the error exponent of the parallel configuration, we need to choose the local Neyman-Pearson test at each relay so that its
local Type II error exponent is close to $g^*_P=-\sup_{\gamma\in\Gamma} \KLD{\P_0^\gamma}{\P_1^\gamma}$.
However, the associated local Type I error cannot fall faster than exponentially, so we can assume it is bounded below by $\delta \exp(-m \epsilon)$, for some $\delta,\epsilon>0$, and for all $m$ large enough.
In that case,  the overall Type I error probability (at the fusion center) is at least $1 - (1 - \delta e^{-m\epsilon})^N$.
We then note that if $N$ increases quickly with $m$ (e.g., $N=m^m$), the Type I error probability approaches 1, and eventually exceeds $\alpha$. Hence, we no longer have an admissible strategy.
Thus, if there is a hope of achieving the optimal exponent $g^*_P$ of the parallel configuration, a more complicated fusion rule will have to be used. \hfill \EndExample
\end{Example}

Our subsequent results will establish that, similar to Example \ref{eg:TwoHeight}, the equalities $g^*=g^*_R=g^*_P$
also hold in Example \ref{eg:InfiniteRelay}. However, Example \ref{eg:InfiniteRelay} shows that in order to achieve this optimal error exponent, we may need to employ nontrivial fusion rules at the fusion center (and for similar reasons at the relay nodes), and various thresholds will have to be properly tuned. The simplicity of the fusion rule in Example \ref{eg:TwoHeight} is not representative.

In our next example, the optimal error exponent is inferior (strictly larger) than that of a parallel configuration.

\begin{Example}\label{eg:SubParallel}
Consider a sequence of 1-bit relay trees with the structure shown
in Figure \ref{fig:SubParallel}.
\begin{figure}[!htb]
\begin{center}
\includegraphics[scale=1]{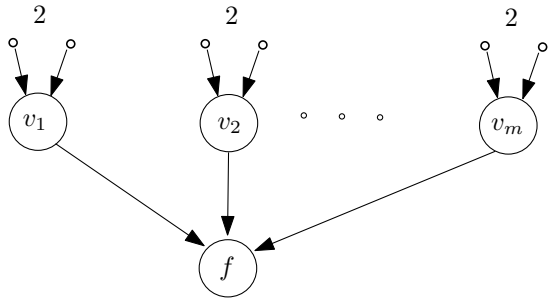}
\caption{A 2-uniform tree, with $m=l_n(f)/2$.}\label{fig:SubParallel}
\end{center}
\end{figure}
Let the observations $X_v$ at the leaves be  i.i.d.\ Bernoulli random variables with parameter $1-p$ under $H_0$, and parameter $p$ under $H_1$, where $1/2 < p < 1$.
Note that
$$g^*_P= \E_0\Big[\log \ddfrac{\P_1^X}{\P_0^X} \Big] =
p \log \frac{1-p}{p} + (1-p) \log \frac{p}{1-p}.
$$

We can identify this relay tree with a parallel configuration involving $m$ nodes,
with each node receiving an independent observation distributed as $\gamma(X_1,X_2)$.
Note that we can restrict
the transmission function $\gamma$ to be the same for all nodes $v_1, ..., v_m$
\cite{Tsi:88}, without loss of optimality.
We have
\begin{align}
\lim_{n\to\infty} \ofrac{m} \log \beta^*(T_n)
 = \min_{\gamma \in \Gamma(2)}
 \sum_{j=0}^1 \P_0\big(\gamma(X_1,X_2)=j\big)
 \log\Big[\frac{\P_1\big(\gamma(X_1,X_2)=j\big)}
 {\P_0\big(\gamma(X_1,X_2)=j\big)}\Big].\label{SubParallelErrExp}
\end{align}
To minimize the right-hand side (R.H.S.) of \eqref{SubParallelErrExp},
we only need to consider a small number of choices for $\gamma$.
If $\gamma(X_1,X_2)=X_1$, we are effectively removing half of the original $2m$ nodes, and the resulting error exponent is $g^*_P/2$, which is inferior to $g^*_P$. Suppose now that $\gamma$ is of the form
$\gamma(X_1,X_2)=0$  iff $X_1=X_2=0$.
Then, it is easy to see, after some calculations (omitted), that
\begin{align*}
\lim_{n\to\infty} \ofrac{m} \log \beta^*(T_n)
& = p^2 \log \frac{(1-p)^2}{p^2} + (1-p^2)\log \frac{1-(1-p)^2}{1-p^2} \\
& > 2 \Big( p \log \frac{1-p}{p} + (1-p) \log \frac{p}{1-p} \Big),
\end{align*}
and
$$
\lim_{n\to\infty} \ofrac{l_n(f)} \log \beta^*(T_n) > p \log \frac{1-p}{p} + (1-p) \log \frac{p}{1-p} = g^*_P.$$

Finally, we need to consider $\gamma$ of the form $\gamma(X_1,X_2)=1$ iff $X_1=X_2=1$. A similar calculation (omitted) shows that the resulting error exponent is again inferior.
We conclude that the relay network is strictly inferior to the parallel configuration,
i.e., $g^*_P < g^*_R$.
An explanation is provided by noting that
this sequence of trees violates a necessary condition,
developed in Section \ref{subsect:NecessaryCondition} for the optimal error exponent to be the same as that of a parallel configuration; see Theorem \ref{thm:MainResults}(\ref{MR:5}). \hfill \EndExample
\end{Example}

A comparison of the results for the previous examples suggests that we have $g^*_P=g^*_R$ (respectively, $g^*_P<g^*_R$) whenever the degree of level 1 nodes increases (respectively, stays bounded) as $n$ increases. That would still leave open the case of networks in which different level 1 nodes have different degrees, as in our next example.

\begin{Example}\label{eg:IncreasingLeaves}
Consider a sequence of $2$-uniform trees of the form shown in Figure \ref{fig:IncreasingLeaves}.
Each node $v_i$, $i=1,...,m$, has $i+1$ leaves attached to it.
We will see that the optimal error
exponent is again the same as for a parallel configuration, i.e., $g^*_R=g^*=g^*_P$.
(cf.\ Theorem \ref{thm:MainResults}(\ref{MR:2})).
\begin{figure}[!htb]
\begin{center}
\includegraphics[scale=0.8]{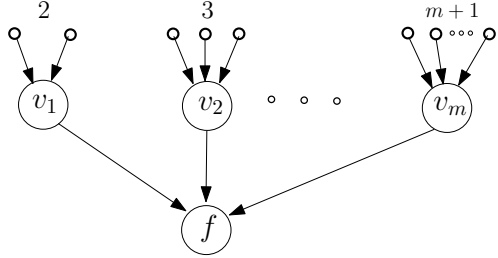}
\caption{A 2-uniform tree, with $l_n(v_i)=i+1$.}\label{fig:IncreasingLeaves}
\end{center}
\end{figure}
\hfill \EndExample
\end{Example}

\subsection{Assumptions}\label{subsect:Assumptions}

In this subsection, we list our assumptions.
Assumptions \ref{assumpt:EquivalentMeasures} and \ref{assumpt:BoundedDivergence} are similar to the assumptions made in the study of the parallel configuration (see \cite{Tsi:88}).

\begin{Assumption}\label{assumpt:EquivalentMeasures}
The measures $\P_0^X$ and $\P_1^X$ are equivalent, i.e., they are absolutely continuous w.r.t.\ each other. Furthermore, there exists some $\gamma\in\Gamma$ such that
$ -\KLD{\P_0^\gamma}{\P_1^\gamma} < 0 < \KLD{\P_1^\gamma}{\P_0^\gamma}$.
\end{Assumption}

\begin{Assumption}\label{assumpt:BoundedDivergence}
$ \E_0 \big[\log^2 \ddfrac{\P_1^X}{\P_0^X}\big] < \infty$.
\end{Assumption}
Assumption \ref{assumpt:BoundedDivergence} implies the following lemma; see \cite{Tsi:88} for a proof.

\begin{Lemma}\label{lemma:BoundedDivergence}
There exists some $a \in (0,\infty)$, such that for all $\gamma\in \Gamma$,
\begin{align*}
& \E_0 \Big[\log^2 \ddfrac{\P_1^\gamma}{\P_0^\gamma}\Big] \leq  \E_0 \Big[\log^2 \ddfrac{\P_1^X}{\P_0^X}\Big] +1 < a,\\
& \E_0 \Big[ \Big| \log \ddfrac{\P_1^\gamma}{\P_0^\gamma} \Big| \Big] < a.
\end{align*}

\end{Lemma}

Given an admissible strategy, and for each node $v \in V_n$,
we consider the log-likelihood ratio of the distribution of $Y_{v}$ (the message sent by $v$) under $H_1$,
w.r.t.\ its distribution under $H_0$,
\begin{align*}
\mathcal{L}_{v,n}(y) = \log \ddfrac{\Pin{v}}{\Pon{v}}(y),
\end{align*}
where $\ud \Pin{v}/ \ud \Pon{v}$ is the Radon-Nikodym derivative of the distribution of $Y_{v}$
under $H_1$ w.r.t.\ that under $H_0$.
If $Y_{v}$ takes values in a discrete set, then
this is just the log-likelihood ratio
$ \log \big(\P_1(Y_{v} = y)/\P_0(Y_{v} = y)\big)$.
For simplicity, we let $L_{v,n} = \mathcal{L}_{v,n}(Y_{v})$ and define
the log-likelihood ratio of the received messages at node $v$ to be
\begin{align*}
S_n(v) = \sum_{u \in C_n(v)} L_{u,n}.
\end{align*}
(Recall that $C_n(v)$ is the set of immediate predecessors of $v$.)

A (1-bit) Log-Likelihood Ratio Quantizer (LLRQ) with threshold $t$ for a non-leaf node $v$,
with $|C_n(v)|=d$,
is a binary-valued function on $\mathcal{T}^d$, defined by
\begin{align*}
{\rm LLRQ}_{d,t} \big(\{y_u : u\in C_n(v)\}\big)
= \left \{ \begin{array}{ll}
0, & {\rm if}\ x \leq t, \\
1, & {\rm if}\ x > t,
\end{array} \right.
\end{align*}
where
\begin{equation}\label{eq:x}
x = \ofrac{l_n(v)} \sum_{u \in C_n(v)} \mathcal{L}_{u,n}(y_u).
\end{equation}
By definition, a node
$v$ that uses a LLRQ ignores its own observation $X_v$ and acts as a relay.
If all non-leaf nodes use a LLRQ, we have a special case of a relay strategy.
We will assume that LLRQs are available choices of transmission functions for all non-leaf nodes.
\begin{Assumption}\label{assumpt:LLRQ}
For all $t\in \Real$ and $ d > 0$, $\mathrm{LLRQ}_{d,t} \in \Gamma(d)$.
\end{Assumption}

As already discussed (cf.\ Eq.\ (\ref{eq:gg})), the optimal performance of a relay tree is always dominated by
that of a parallel configuration with the same number of leaves, i.e., $g^*_P \leq  g^*_R$.
In Section \ref{sect:OptimalError}, we find sufficient conditions under which the equality  $g^*_R = g^*_P$ holds. Then, in
Section \ref{subsect:NecessaryCondition}, we look into necessary conditions for  this to be the case.
It turns out that
non-trivial necessary conditions for the equality $g^*_R = g^*_P$ to hold are, in general, difficult to obtain, because they depend on the nature of the transmission functions available to the sensors. For example, if the sensors are allowed to simply forward undistorted all of the messages that they receive,  then the equality $g^*_R = g^*_P$ holds trivially.  Hence, we need to impose some restrictions on the set of transmission functions available, as in the assumption that follows.

\begin{Assumption}\label{assumpt:LossOfDivergence}\
\begin{enumerate}[(a)]
\item There exists a $n_0 \geq 1$ such that for all $n\geq n_0$, we have $l_n(v) > 1$ for all $v$ in the set $B_n$ of
nodes whose immediate predecessors are all leaves.
\item
Let $X_1,X_2,\ldots$ be i.i.d.\ random variables under either hypothesis $H_j$, each with distribution $\P_j^{X}$.
For $k > 1$,
$\gamma_0 \in \Gamma(k)$, and $\gamma_i \in \Gamma$, $i=1,\ldots,k$, let $\xi =(\gamma_0,\ldots,\gamma_k)$.
We also let $\nu_j^\xi$ be the distribution of $\gamma_0(\gamma_1(X_1),\ldots,\gamma_k(X_k))$ under hypothesis $H_j$.
We assume that
\begin{align}\label{LossOfDivergence}
g^*_P <
\inf_{\xi\in \Gamma(k)\times\Gamma^k} \ofrac{k} \E_0\Big[\log \ddfrac{\nu_1^\xi}{\nu_0^\xi}\Big],
\end{align}
for all $k > 1$.
\end{enumerate}
\end{Assumption}

Assumption \ref{assumpt:LossOfDivergence} holds in most cases of interest.
Part (a) results in no loss of generality: if in a relay tree we have $l_n(v)=1$ for some $v\in B_n$, we can remove the predecessor of $v$, and treat $v$ as a leaf. Regarding part (b),
it is easy to see that the left-hand side (L.H.S.) of (\ref{LossOfDivergence})
is always less than or equal to the R.H.S.,
hence we have only excluded those cases where (\ref{LossOfDivergence}) holds with equality. We are essentially assuming that when
the
messages $\gamma_1(X_1),\ldots,\gamma_k(X_k)$ are summarized (or quantized) by $\gamma_0$, there is some loss of information, as measured by the associated KL divergences.

\subsection{Main Results}\label{subsect:MainResults}

In this section, we collect and summarize our main results. The asymptotic proportion of nodes that are leaves, defined by
$$z=\liminf\limits_{n\to\infty} \frac{l_n(f)}{n},$$
plays a critical role.

\begin{Theorem}\label{thm:MainResults}
Consider a sequence of trees, $(T_n)_{n\geq1}$, of bounded height.
Suppose that Assumptions \ref{assumpt:EquivalentMeasures}-\ref{assumpt:LLRQ} hold. Then,
\begin{enumerate}[(i)]
\item\label{MR:1}
$g^*_P \leq g^*_R < 0$ and $-\KLD{\P_0^X}{\P_1^X} \leq g^* \leq z g^*_R < 0$.
\item \label{MR:2}
If $z=1$, then $g^*_P = g^* = g^*_R$.
\item \label{MR:4}
If the number of non-leaf nodes is bounded, or if
$\min_{v\in B_n}l_n(v)\to\infty$,
then
$g^*_P=g^*=g^*_R$.
\item \label{MR:5}
If Assumption \ref{assumpt:LossOfDivergence} also holds, we have $g^*_R=g^*_P$ iff $z=1$.
\end{enumerate}
\end{Theorem}

Note that part (\ref{MR:1}) follows from (\ref{eq:gg}), except for the strict negativity of the error exponents, which is established in Proposition \ref{prop:OtherTypes}.
Part (\ref{MR:2}) is proved in Proposition \ref{prop:NeymanPearsonSuff}.
Part (\ref{MR:4}) is proved in Corollary \ref{cor:Sufficient}.
(Recall that $B_n$ is the set of non-leaf nodes all of whose immediate predecessors
are leaves.)
Part (\ref{MR:5}) is proved in
Proposition \ref{prop:NecessaryAndSufficient}.
One might also have expected a result asserting that $g^*_P\leq g^*$. However, this is not true without additional assumptions, as will be discussed in Section \ref{subsect:NecessaryCondition}.

\section{Error Bounds for $h$-Uniform Relay Trees}\label{sect:SimpleRelay}

In this section, we consider a 1-bit $h$-uniform relay tree, in which all relay nodes at level $k$
use a LLRQ with a common threshold $t_k$.
We wish to develop upper bounds for the error probabilities at the various nodes.
We do this recursively, by moving along the levels of the tree, starting from the leaves.
Given bounds on the error probabilities associated with the messages received by a node, we develop a bound on the log-moment generating function at that node (cf.\ Eq.\ (\ref{eq:mgf})), and then use the standard Chernoff bound technique to develop a bound on the error probability for the message sent by that node (cf.\ Eq.\ (\ref{eq:lstar})).

Let $t\tn{k} = (t_1, t_2, \ldots, t_k)$, for $k\geq1$, and $t\tn{0}=\emptyset$. For $j=0,1$, $k\geq1$, and $\lambda\in\Real$, we
define recursively
\begin{align}
& \Lambda_{j,0}(\gamma; \lambda) =\Lambda_{j,0}(\gamma,\emptyset; \lambda)
= \log \E_j \Big[ \Big(\ddfrac{\P_{1}^\gamma}{\P_{0}^\gamma}\Big)^\lambda  \Big],\nonumber\\
& \Lambdahs{j}{k} = \sup_{\lambda \in \Real} \big \{ \lambda t_k - \Lambda_{j,k-1}(\gamma,t\tn{k-1}; \lambda) \big \}, \label{eq:lstar}\\
& \Lambda_{j,k}(\gamma,t\tn{k}; \lambda) = \max \big \{ - \Lambdahs{1}{k} (j + \lambda), \Lambdahs{0}{k}(j-1+\lambda) \big \}. \label{eq:mgf}
\end{align}
The operation in (\ref{eq:lstar}) is known as the
Fenchel-Legendre transform of $\Lambda_{j,k-1}(\gamma,t\tn{k-1}; \lambda)$ \cite{DemZei:98}.
We will be interested in the case where
\begin{align}
& -\KLD{\P_0^\gamma}{\P_1^\gamma}< 0 < \KLD{\P_1^\gamma}{\P_0^\gamma}, \label{barx}\\
& t_1 \in \big(-\KLD{\P_0^\gamma}{\P_1^\gamma}, \KLD{\P_1^\gamma}{\P_0^\gamma}\big), \label{t1interval}\\
& t_k \in \big(- \Lambdahs{1}{k-1}, \Lambdahs{0}{k-1}\big), \textrm{ for $1 < k \leq h$}. \label{tkinterval}
\end{align}

We now provide an inductive argument to show that the above requirements on the thresholds $t_k$ are feasible.
From Assumption \ref{assumpt:EquivalentMeasures}, there exists a $\gamma\in \Gamma$ that satisfies
(\ref{barx}), hence the constraint (\ref{t1interval}) is feasible.
Furthermore, the $\Lambdahs{j}{1}$ are large deviations rate functions and are therefore positive when $t_1$ satisfies (\ref{t1interval}) \cite{DemZei:98}.
Suppose now that $k>1$ and that $\Lambdahs{j}{k-1} > 0$.
From (\ref{eq:mgf}), $\Lambda_{j,k-1}(\gamma,t\tn{k-1};\lambda)$ is the maximum of two linear functions
of $\lambda$ (see Figure \ref{fig:MaxLinear}). Taking the Fenchel-Legendre transform, and since
$t_k$ satisfies (\ref{tkinterval}), we obtain $\Lambda_{j,k}^*(\gamma,t\tn{k}) >0$, which completes the induction.

\begin{figure}[!htb]
\begin{center}
\includegraphics[scale=1]{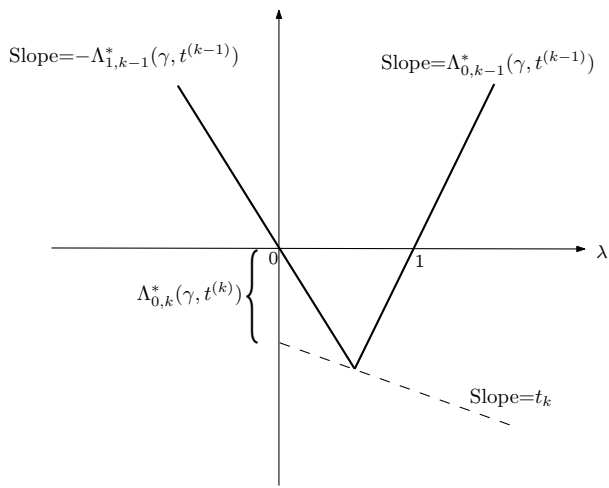}
\caption{Typical plot of $\Lambda_{0,k-1}(\gamma,t\tn{k-1};\lambda)$, $k\geq2$.}\label{fig:MaxLinear}
\end{center}
\end{figure}

From the definitions of $\Lambda_{j,k}$ and $\Lambda_{j,k}^*$, the following relations can be established. The proof consists of straightforward algebraic manipulations and is omitted.

\begin{Lemma}\label{lemma:LambdaRelationship}
Suppose that $\gamma\in\Gamma$ satisfies (\ref{barx}), and $t\tn{h}$ satisfies (\ref{t1interval})-(\ref{tkinterval}). For $k\geq1$, we have
\begin{align*}
\Lambdahs{1}{k} & = \Lambdahs{0}{k} - t_k.
\end{align*}
Furthermore, the supremum in (\ref{eq:lstar}) is achieved at some $\lambda\in (-1,0)$ for $j=1$, and $\lambda\in (0,1)$ for $j=0$.
For $k\geq 2$, we have
\begin{align*}
\Lambdahs{1}{k} & = \frac{\Lambdahs{1}{k-1}(\Lambdahs{0}{k-1} - t_k)}{\Lambdahs{0}{k-1} + \Lambdahs{1}{k-1}}, \\
\Lambdahs{0}{k} & = \frac{\Lambdahs{0}{k-1}(\Lambdahs{1}{k-1} + t_k)}{\Lambdahs{0}{k-1} + \Lambdahs{1}{k-1}}.
\end{align*}
\end{Lemma}
\vspace{5pt}

Proposition \ref{prop:UpperBddErr} below, whose proof is provided in the Appendix,
will be our main tool in obtaining upper bounds on error probabilities.
It shows that the Type I and II error exponents are essentially upper bounded by $-\Lambdahs{0}{h}$ and $-\Lambdahs{1}{h}$ respectively.
Recall that $p_n(v)$ is the total number of predecessors of $v$, $l_n(v)$ is the number of leaves in the sub-tree rooted at $v$, and $B_n$ is the set of nodes all of whose immediate predecessors are leaves.

\begin{Proposition}\label{prop:UpperBddErr}
Fix some $h\geq 1$, and
consider a sequence of trees $(T_n)_{n\geq1}$ such that for
all $n\geq n_0$,  $T_n$ is $h$-uniform. Suppose that Assumptions \ref{assumpt:EquivalentMeasures}-\ref{assumpt:BoundedDivergence} hold. Suppose that,
for every $n$, every leaf uses the same transmission function $\gamma\in\Gamma$, which satisfies (\ref{barx}),
and that every level $k$ node ($k\geq1$) uses a LLRQ with threshold $t_k$, satisfying (\ref{t1interval})-(\ref{tkinterval}).
\begin{enumerate}[(i)]
\item\label{it:general}
For all nodes $v$ of level $k\geq1$ and for all $n\geq n_0$, we have
\begin{align*}
\ofrac{l_n(v)} \log \P_1 \Big ( \frac{S_n(v)}{l_n(v)} \leq t_k \Big ) & \leq -\Lambdahs{1}{k} + \frac{p_n(v)}{l_n(v)}-1,\\
\ofrac{l_n(v)} \log \P_0 \Big ( \frac{S_n(v)}{l_n(v)} > t_k \Big ) & \leq -\Lambdahs{0}{k} + \frac{p_n(v)}{l_n(v)}-1.
\end{align*}
\item\label{it:N}
Suppose that for all $n\geq n_0$ and all $v \in {B}_n$, we have $l_n(v) \geq N$.
Then, for all $n\geq n_0$, we have
\begin{align*}
\ofrac{l_n(f)} \log \P_1 \Big ( \frac{S_n(f)}{l_n(f)} \leq t_h \Big ) & \leq -\Lambdahs{1}{h} + \frac{h}{N},\\
\ofrac{l_n(f)} \log \P_0 \Big ( \frac{S_n(f)}{l_n(f)} > t_h \Big ) & \leq -\Lambdahs{0}{h} + \frac{h}{N}.
\end{align*}
\end{enumerate}
\end{Proposition}

\section{Optimal Error Exponent}\label{sect:OptimalError}

In this section, we show that the Type II error probability in a sequence of bounded height trees falls exponentially
fast with the number of nodes. We derive sufficient conditions for the error exponent to be the same as that of a parallel configuration.
We show that if almost all of the nodes are leaves, i.e., $z=1$, then $g^*_P=g^*=g^*_R$. The condition $z=1$ is also equivalent to another condition that requires that the proportion of leaves attached to bounded degree nodes
vanishes asymptotically.
We also show that under some additional mild assumptions, this sufficient condition is necessary. We start with some graph-theoretic preliminaries.

\subsection{Properties of Trees.}
In this section, we define various quantities associated with a tree, and derive a few elementary relations that will be used later.

Recall that
$B_n$ is the set of non-leaf nodes all of whose predecessors are leaves.
(For an $h$-uniform tree, $B_n$ is the set of all level 1 nodes.)
For $N >0$, let
\begin{align}\label{F_Nn}
F_{N,n} & = \{ v \in B_n : l_n(v) \leq N \},
\qquad  F^c_{N,n}  = \{ v \in B_n : l_n(v) > N \},
\end{align}
and
\begin{align}\label{q_Nn}
q_{N,n} = \ofrac{l_n(f)} \sum_{v \in F_{N,n}} l_n(v),
\end{align}
where the sum is taken to be zero if the set $F_{N,n}$ is empty. Let
$q_N = \limsup\limits_{n\to\infty} q_{N,n}$.
For a sequence of $h$-uniform trees,
this is the asymptotic proportion of leaves that belong to ``small'' subtrees in the network.

It turns out that it is easier to work with $h$-uniform trees. For this reason, we show how to transform any
tree of height $h$ to an $h$-uniform tree.
\vspace{6pt}

\noindent
\textbf{Height Uniformization Procedure.}
Consider a tree $T_n=(V_n,E_n)$ of height $h$, and a node $v$ that has at least one leaf as an immediate predecessor ($v \in A_n$).
Let $D_n$ be the set of leaves that are immediate predecessors of $v$, and whose
paths to the fusion center $f$ are of length $ k < h$.
Add $h-k$ nodes, $\{u_j: j=1,\ldots,h-k\}$, to $V_n$; remove the edges $(u,v)$, for all $u \in D_n$;
add the edges $(u_1,v)$, and $(u_{j+1},u_j)$, for $j=1,\ldots,h-k-1$; add the edges $(u,u_{h-k})$, for all $u \in D_n$.
This procedure is repeated for all $v \in A_n$. The resulting tree is $h$-uniform. $\square$
\vspace{6pt}

The height uniformization procedure essentially adds more nodes to the network, and
re-attaches some leaves, so that the path from every leaf has exactly $h$ hops.
Let $(T'_n=(V_n', E_n') )_{n\geq 1}$ be the new sequence of $h$-uniform trees obtained from $(T_n)_{n\geq 1}$, after applying the uniformization procedure. (We are abusing notation here in that $T'_n$ typically does not have $n$ nodes,
nor is the sequence $|V_n'|$ increasing.)
Regarding notation,
we adopt the convention that quantities marked with a prime are defined with respect to $T_n'$.

Note that $l_n'(f) = l_n(f)$. For the case of a relay network, it is seen that any function of the observations at the leaves that can be computed in $T_n'$ can also be computed in $T_n$. Thus, the detection performance of
$T_n'$ is no better than that of $T_n$. Hence, we obtain
\begin{align}\label{eq:un}
g^*_R \leq \limsup_{n\to\infty} \ofrac{l_n'(f)} \log \beta^*(T_n').
\end{align}
Therefore, any upper bound derived for $h$-uniform trees, readily translates to an upper bound for general trees. On the other hand, the coefficients $q_N$ for the $h$-uniform trees $T'_n$ (to be denoted by $q'_N$) are different from the coefficients $q_N$ for the original sequence $T_n$. They are related as follows. The proof is given in the Appendix.

\begin{Lemma}\label{lemma:Uniformizedq_N}
For any $N, M >0$, we have
\begin{align*}
q'_N \leq h(N q_M + N/M).
\end{align*}
In particular, if $q_N=0$ for all $N>0$, then $q'_N=0$ for all $N>0$.
\end{Lemma}

It turns out that the condition $z=1$ is equivalent to the condition $q_N=0$ for all $N>0$.
The proof is provided in the Appendix.
\begin{Lemma}\label{lemma:Equivalence}
We have
$z=1$ iff $q_N=0$ for all $N>0$.
\end{Lemma}

\subsection{An Upper Bound}
In this section, we develop an upper bound on the Type II error probabilities, which takes into account some qualitative properties of the sequence of trees, as captured by $q_N$.

\begin{Lemma}\label{lemma:upperbound2}
Consider an $h$-uniform sequence of trees $(T_n)_{n\geq1}$, and suppose that Assumptions \ref{assumpt:EquivalentMeasures}-\ref{assumpt:LLRQ}
hold. For every $\epsilon > 0$,
there exists some $N$ such that
$$g^*_R \leq (1-q_N) (g^*_P +\epsilon).$$
\end{Lemma}
\begin{proof}
If $g^*_P +\epsilon \geq 0$, there is nothing to prove, since $q_N \leq 1$ and $g_R^*\leq 0$. Suppose that $g^*_P +\epsilon < 0$.
Choose $\gamma\in \Gamma$ such that
\begin{align*}
-\KLD{\P_0^\gamma}{\P_1^\gamma}
& \leq -\sup_{\gamma'\in\Gamma} \KLD{\P_0^{\gamma'}}{\P_1^{\gamma'}} + \frac{\epsilon}{2} =g^*_P + \frac{\epsilon}{2} < 0.
\end{align*}

Let $t_k = t = -\KLD{\P_0^\gamma}{\P_1^\gamma} + \epsilon/2 \leq g_p^*+\epsilon$, for $k=1,\ldots,h$, and note that
\begin{align}
-\KLD{\P_0^\gamma}{\P_1^\gamma} < t <0. \label{Rangeoft1}
\end{align}
Because of (\ref{Rangeoft1}), we have $\Lambda^*_{0,1}(\gamma,t^{(1)}) > 0$. Furthermore, using Lemma \ref{lemma:LambdaRelationship}, $\Lambda^*_{1,1}(\gamma,t^{(1)})= \Lambda^*_{0,1}(\gamma,t^{(1)})- t > -t$.
Now let $k\geq 2$, and suppose that $\Lambdahs{1}{k-1} > -t$ and $\Lambdahs{0}{k-1} > 0$.
From Lemma \ref{lemma:LambdaRelationship},
\begin{align*}
\Lambdahs{0}{k} & = \frac{\Lambdahs{0}{k-1}(\Lambdahs{1}{k-1} + t)}{\Lambdahs{0}{k-1} + \Lambdahs{1}{k-1}} > 0,
\end{align*}
and
\begin{align*}
\Lambdahs{1}{k} & = \Lambdahs{0}{k} - t_k = \Lambdahs{0}{k} - t > -t.
\end{align*}
Hence, by induction, $t_k$ satisfies (\ref{t1interval})-(\ref{tkinterval}),
so that Proposition \ref{prop:UpperBddErr} can be applied.

Choose $N$ sufficiently large so that $h/N < \Lambdahs{0}{h}$.
If $q_N = 1$, the claimed result holds trivially. Hence, we assume that $q_N \in [0,1)$.
In this case,
for $n$ sufficiently large, there exists at least one node in $B_n$ so that $l_n(v) > N$.
We remove all nodes $v \in B_n$ with $l_n(v) \leq N$, and their immediate predecessors.
Then, we remove all level 2 nodes $v$ that no longer have any predecessors,
and so on. In this way, we obtain an $h$-uniform subtree of $T_n$, to be denoted by $T_n''$.
(Quantities marked with double primes are defined w.r.t.\ $T_n''$.)
We have $l_n''(v) > N$ for all $v\in B_n''$, and $l_n''(f) = \sum_{v \in F_{N,n}^c} l_n(v) =l_n(f) (1-q_{N,n})$.
Consider the following relay strategy on the tree $T_n''$.
(Since this is a subtree of $T_n$, this is also a relay strategy for the tree $T_n$, with some nodes remaining idle.)
The leaves transmit with transmission function $\gamma$, and the other nodes use a 1-bit LLRQ with threshold $t$.
(Note that in the definition
(\ref{eq:x}) of the normalized log-likelihood ratio, the denominator $l_n(v)$ now becomes $l''_n(v)$.)

We first show that the strategy just described is admissible.
We apply part (\ref{it:N}) of Proposition \ref{prop:UpperBddErr} to $T_n''$, to obtain
\begin{align*}
& \limsup_{n \to \infty} \ofrac{l_n(f)} \log \P_0(Y_{f}=1) \\
& = \limsup_{n \to \infty} \frac{l_n''(f)}{l_n(f)} \cdot \ofrac{l_n''(f)} \log \P_0(Y_{f}=1) \\
& \leq (1-q_N) \limsup_{n \to \infty} \ofrac{l_n''(f)} \log \P_0\Big( \frac{S_n(f)}{l_n''(f)} > t\Big)\\
& \leq (1-q_N) \big( -\Lambdahs{0}{h} + \frac{h}{N} \big) < 0,
\end{align*}
hence $\P_0(Y_{f}=1) \leq \alpha$, when $n$ is sufficiently large.

To bound the Type II error probability, we use
Proposition \ref{prop:UpperBddErr} and Lemma \ref{lemma:LambdaRelationship}, to obtain
\begin{align*}
g^*_R
& \leq \limsup_{n \to \infty} \ofrac{l_n(f)} \log \beta^*(T''_n) \nonumber\\
& \leq (1-q_N) \limsup_{n \to \infty} \ofrac{l_n''(f)} \log \P_1 \Big( \frac{S_n(f)}{l_n''(f)} \leq t \Big) \nonumber\\
& \leq  (1-q_N) \big( - \Lambdahs{1}{h} + \frac{h}{N} \big) \\
& = (1-q_N) \big( t - \Lambdahs{0}{h} + \frac{h}{N} \big) \\
& \leq (1-q_N) t\\
& \leq (1-q_N) \big( g_P^* + \epsilon \big).
\end{align*}
This proves the lemma.
\end{proof}

\subsection{Exponential decay of error probabilities}\label{subsect:ExponentialDecay}

We now establish that Type II error probabilities decay exponentially. The bounded height assumption is crucial for this result. Indeed, for the case of a tandem configuration, the exponential decay property does not seem to hold.

\begin{Proposition}\label{prop:OtherTypes}
Consider a sequence of trees of height $h$, and let
Assumptions \ref{assumpt:EquivalentMeasures}-\ref{assumpt:LLRQ} hold.
Then,
\begin{align*}
-\infty < g^*_P \leq g^*_R <0 \qquad \textrm{and}\qquad
-\infty < - \KLD{\P_0^X}{\P_1^X} \leq g^* <0.
\end{align*}
\end{Proposition}

\begin{proof}
The lower bounds on $g^*_R$ and $g^*$ follow from (\ref{eq:gg}).
Note that $g^*_P$ cannot be equal to $-\infty$ because it cannot be better than the error exponent of a parallel configuration in which all the observations are provided uncompressed to the fusion center. The error exponent in the latter case is $-\KLD{\P_0^X}{\P_1^X}$, by Stein's Lemma, and is finite as a consequence of Assumption \ref{assumpt:BoundedDivergence}.

It remains to show that the optimal error exponents are negative. Every tree of height $h$ satisfies $n \leq l_n(f) h+1$.
From (\ref{eq:gg}), we obtain $g^* \leq g^*_R/h$.
Therefore, we only need to show that $g^*_R<0$.
As discussed in connection to (\ref{eq:un}), we can restrict attention to a sequence of $h$-uniform trees.

We use induction on $h$. If $h=1$, we have a parallel configuration and the result follows from \cite{Tsi:88}.
Suppose that the result is true for all sequences of
$(h-1)$-uniform trees. Consider now a sequence of $h$-uniform trees.
Let $\epsilon>0$ be such that $g^*_P+\epsilon <0$.
From Lemma \ref{lemma:upperbound2}, there exists some $N$ such that
$g^*_R \leq (1-q_N) (g^*_P +\epsilon)$. If $q_N<1$, we readily obtain the inequality
$g^*_R<0$.

Suppose now that $q_N=1$. We
only need to consider a sequence $(n_k)_{k\geq1}$ such that
$\lim\limits_{k\to\infty} q_{N,n_k} = 1$.
Using the inequality (\ref{eq:F}), we have
$$
\frac{|F_{N,n_k}|}{l_{n_k}(f)}
 \geq \frac{q_{N,n_k}}{N},$$
and
\begin{equation}
\liminf_{k \to \infty} \frac{|F_{N,n_k}|}{l_{n_k}(f)}
\geq \frac{1}{N}.\label{ProportionBdd}
\end{equation}

For each node $v \in B_n$, we remove all of its immediate predecessors (leaves) except for one, call it $u$. The leaf $u$ transmits $\gamma(X_u)$ to its immediate successor $v$. Since node $v$ receives only a single message, it just  forwards it to its immediate successor. The resulting performance is the same
as if the nodes $v$ in $B_n$ were making a measurement $X_v$ and transmitting $\gamma(X_v)$ to their successor.
This is equivalent to deleting all the leaves of $T_n$ to form a new tree, $T''_n$, which is $(h-1)$-uniform.
The above argument shows that
$\beta^*(T_{n_k}) \leq \beta^*(T''_{n_k})$.

We have
$l_{n_k}''(f) = |B_{n_k}|$ and from (\ref{ProportionBdd}),
\begin{align*}
\liminf_{k \to \infty} \frac{|B_{n_k}|}{l_{n_k}(f)} \geq \liminf_{k\to\infty} \frac{|F_{N,n_k}|}{l_{n_k}(f)}
& \geq \frac{1}{N}.
\end{align*}
Therefore,
\begin{align*}
\limsup_{k \to \infty} \ofrac{l_{n_k}(f)} \log \beta^*(T_{n_k})
& \leq \ofrac{N} \limsup_{k \to \infty} \ofrac{l_{n_k}''(f)} \log \beta^*(T_{n_k}'').
\end{align*}
By the induction hypothesis, the right-hand side in the above inequality is negative and the proof is complete.
\end{proof}

\subsection{Sufficient Conditions for Matching the Performance of the Parallel Configuration}\label{subsect:SufficientConditions}

We are now ready to prove the main result of this section. It shows that when $q_N=0$ for all $N>0$, or equivalently when
$z=1$ (cf.\ Lemma \ref{lemma:Equivalence}), bounded height tree networks match the performance of the parallel configuration.

\begin{Proposition}\label{prop:NeymanPearsonSuff}
Consider a sequence of trees of height $h$ in which $z=1$, or equivalently $q_N=0$ for all $N>0$.
Suppose that Assumptions \ref{assumpt:EquivalentMeasures}-\ref{assumpt:LLRQ} hold.
Then,
\begin{align*}
g^*_P=g^*=g^*_R.
\end{align*}
Furthermore, if the sequence of trees is $h$-uniform, the optimal error exponent does not change
even if we restrict to relay strategies in which every leaf uses the same transmission function and all other nodes use a 1-bit LLRQ with the same threshold.
\end{Proposition}
\begin{proof}
We have shown $g^*_P\leq g^*_R$ in (\ref{eq:gg}). We now prove that
$g^*_R\leq g^*_P$.
As already explained, there is no loss in generality in assuming that the sequence of trees is $h$-uniform
(by performing the height uniformization procedure, and using Lemma \ref{lemma:Uniformizedq_N}).

For any $\epsilon > 0$, Lemma \ref{lemma:upperbound2} yields
$$
g^*_R\leq g^*_P+\epsilon.
$$
Letting $\epsilon\to 0$, we obtain $g^*_R\leq g^*_P$, hence $g^*_R=g^*_P$.
From (\ref{eq:gg}) with $z=1$, we obtain $g^*\leq g^*_R = g^*_P$.

We now show that $g^* \geq g^*_P$.
Consider a tree with $n$ nodes, $l_n(f)$ of which are leaves.
We will compare it with another sensor network in which $l_n(f)$ nodes $v$ transmit
a message $\gamma_v(X_v)$ to the fusion center and $n-l_n(f)-1$ nodes transmit their raw observations to the fusion center. The latter network can simulate the
original network, and therefore its optimal error exponent is at least as good.
By a standard argument (similar to the one in Proposition \ref{prop:NecessaryConditions} below), the optimal error exponent in the latter network can be shown to be greater than or equal to
$$
\limsup_{n\to\infty} \frac{l_n(f)}{n} g^*_P + \limsup_{n\to\infty}-\frac{n-l_n(f)-1}{n} \KLD{\P_0^X}{\P_1^X} = g^*_P,
$$
hence concluding the proof.
\end{proof}

Fix an $\epsilon \in (0, -g_P^*)$.
For any tree sequence with $z=1$, we can perform the height uniformization
procedure to obtain an $h$-uniform sequence of trees. In practice, this height
uniformization procedure may be performed virtually at each node, so that
the tree sequence simulates a $h$-uniform tree sequence.
A simple strategy on the height uniformized tree sequence
that $\epsilon$-achieves the optimal error exponent is a relay strategy in which:
\begin{enumerate}[(i)]
\item all leaves transmit with the same transmission function $\gamma\in\Gamma$, where
$\gamma$ is chosen such that $-\KLD{\P_0^\gamma}{\P_1^\gamma}\leq g_P^* + \epsilon/2$;
\item all other nodes use 1-bit LLRQs with the same threshold $t =-\KLD{\P_0^\gamma}{\P_1^\gamma} + \epsilon/2$.
\end{enumerate}
Lemmas \ref{lemma:Uniformizedq_N} and \ref{lemma:Equivalence}, and the proof of Lemma \ref{lemma:upperbound2}
shows that this
relay strategy $\epsilon$-achieves the optimal error exponent $g_R^*=g^*=g_P^*$.
This also shows that there is no loss in optimality even if we restrict the relay nodes
to use only 1-bit LLRQs. This may be useful in situations where the nodes are simple, low-cost devices.

Proposition \ref{prop:NeymanPearsonSuff}
provides sufficient conditions for a sequence of trees to
achieve the same error
exponent as the parallel configuration. We note a few special cases in which these sufficient conditions are satisfied. The first one is the case where there is a finite bound on the number of nodes that are not leaves.
In that case, $z$ is easily seen to be 1.
This is consistent with the conclusion of Example \ref{eg:TwoHeight}, where a simpler argument was used. The second is the more general case where nodes in $B_n$ are attached to a growing number of leaves, which implies that $q_N=0$ for all $N>0$.

\begin{Corollary}\label{cor:Sufficient}
Suppose that Assumptions \ref{assumpt:EquivalentMeasures}-\ref{assumpt:LLRQ} hold.
Suppose further that either
of the following conditions holds:
\begin{enumerate}[(i)]
\item \label{c:a}There is a finite bound on the number of nodes that are not leaves.
\item \label{c:b} We have $\min_{v\in B_n} l_n(v)\to \infty$.
\end{enumerate}
Then, $g^*_P=g^*=g^*_R$.
\end{Corollary}

The above corollary can be applied to Example \ref{eg:InfiniteRelay}. In that example,
every level 1 node has $m$ leaves attached to it, with $m$ growing large as $n$ increases.
Therefore, the tree network satisfies condition \eqref{c:b} in Corollary \ref{cor:Sufficient}, and
the optimal error exponent is $g^*=g^*_R=g^*_P$. In this case, even if the number $N$ of
level 1 nodes grows much faster than $m$, we still achieve the same error exponent as the
parallel configuration. The above proposed strategy, in which every leaf uses the same
transmission function, and every node uses the same LLRQ, will nearly achieve the
optimal performance.

We are now in a position to determine the optimal error exponent in Example \ref{eg:IncreasingLeaves}.

\noindent
{\bf Example \ref{eg:IncreasingLeaves}, revisited:}
Recall that in Example \ref{eg:IncreasingLeaves}, every $v_i \in B_n$ has
$i+1$ of predecessors. It is easy to check that $z=1$. From Proposition \ref{prop:NeymanPearsonSuff},
the optimal error exponent is the same as that for the parallel configuration.
\hfill \EndExample

\subsection{Discussion of the Sufficient Conditions}\label{subsect:DiscussionSufficient}

Proposition \ref{prop:NeymanPearsonSuff} is unexpected as it establishes that
the performance of a tree
possessing certain qualitative properties is comparable to that of the parallel configuration.
Furthermore, the optimal performance is obtained even if we restrict the non-leaf nodes to use 1-bit LLRQs.
At first sight, it might appear
intuitive that if the leaves dominate in a relay tree ($z=1$),
then the tree should always have the same performance as
a parallel configuration. However, this intuition is misleading, as this is not the case
for a Bayesian formulation, in which both the Type I and II error probabilities are required
to decay at the same rate, is involved. To see this, consider the 2-uniform tree in Figure \ref{fig:TwoHeight},
where every node is constrained to sending 1-bit messages. Suppose we are given nonzero prior
probabilities $\pi_0$ and $\pi_1$ for the hypotheses $H_0$ and $H_1$.
Instead of
the Neyman-Pearson criterion, suppose that
we are interested in minimizing the error exponent
\begin{align*}
\limsup_{n\to\infty} \ofrac{l_n(f)} \log P_e^*,
\end{align*}
where $P_e^*$ is the minimum of the error probability $\pi_0\P_0(Y_f=1) + \pi_1\P_1(Y_f=0)$, optimized
over all strategies.
It can be shown that to obtain the optimal error exponent,
we only need to consider the following two fusion rules:
(a) the fusion center declares $H_0$ iff both $v_1$ and $v_2$ send a 0, or (b) the fusion center declares $H_1$ iff both $v_1$ and $v_2$ send a 1. Then, using the results in
\cite{TayTsiWin:C07b}, the optimal error exponent for this tree network is strictly worse
than that for the parallel configuration.
Similarly, if we constrain the Type I error in the Neyman-Pearson criterion
to decay faster than a predetermined rate, it can be shown that the optimal Type II error exponent for a tree network can be strictly worse than that of a parallel configuration.

Note that the bounded height assumption is essential in proving $g^*=g^*_R=g^*_P$, when
$z=1$. Although our technique can be extended to include those tree sequences whose height grows very
slowly compared to $n$ (on the order of $\log|\log(n/l_n(f)-1)|$), we have not been able to
find the optimal error exponent for the general case of unbounded height.
As noted before, in a tandem network, the Bayesian
error probability decays sub-exponentially fast \cite{TayTsiWin:C07c}.
The proof of Proposition 2 in \cite{TayTsiWin:C07c} involves the construction of
a tree network, with unbounded height, and in which $z=1$. In that proof, it is also
shown that such a network has a sub-exponential rate of error decay.
We conjecture that this is also the
case for the Neyman-Pearson formulation.

In summary, for a tree network to achieve the same Type II error exponent as a
parallel configuration, we require that the tree sequence have a bounded height,
satisfy the condition $z=1$, and that
the error criterion be the Neyman-Pearson criterion.
Without any one of these three conditions, our results no longer hold.

\subsection{A Necessary Condition for Matching the Performance of the Parallel Configuration}\label{subsect:NecessaryCondition}

In this section, we establish necessary conditions under which a sequence of relay trees with bounded height
performs as well as a parallel configuration.
As noted in Section \ref{subsect:Assumptions}, any necessary conditions generally depend on the type of transmission
functions available to the relay nodes. However, under an additional condition (Assumption \ref{assumpt:LossOfDivergence}),
the sufficient condition for $g^*_R=g^*_P$ in Proposition \ref{prop:NeymanPearsonSuff} is also necessary.

\begin{Proposition}\label{prop:NecessaryConditions}
Suppose that Assumptions \ref{assumpt:EquivalentMeasures}, \ref{assumpt:BoundedDivergence} and \ref{assumpt:LossOfDivergence} hold, and $h \geq 2$.
If there exists some $N>0$ such that $q_N > 0$ (equivalently, $z < 1$), then
$g^*_P < g^*_R$.
\end{Proposition}
\begin{proof}
Fix some $N>0$ and suppose that $q_N>0$. Given a tree $T_n$, we construct a
new tree $T''_n$, as follows.
We remove all nodes other than the leaves and the nodes in $F_{N,n}$.
For all the leaves $u$ that are not immediate predecessors of some $v \in F_{N,n}$,
we let $u$ transmit its message directly to the fusion center.
We add new edges $(v,f)$, for each $v\in F_{N,n}$. This gives us a tree $T''_n$ of height 2, with
$l_n''(f) = l_n(f)$ and $q''_{N} = q_{N}$. The latter tree $T''_n$ can simulate the tree $T_n$, hence
the optimal error exponent
associated with the sequence $(T_n)_{n\geq 1}$ is bounded below by the optimal error exponent associated with the sequence $(T''_n)_{n\geq 1}$.
Therefore, without loss of generality, we only need to prove the proposition for a sequence of trees of height 2,
and in which $F_{N,n}=B_n$, for some $N>0$ such that $q_N >0$; we henceforth assume that this is the case. The rest of the argument is similar to the proof of Stein's Lemma in Lemma 3.4.7 of
\cite{DemZei:98}.
Suppose that a particular admissible relay strategy has been fixed, and
let $\beta_n$ be the associated Type II error probability. Let $\lambda_n = \E_0[{S}_n(f)]/l_n(f)$.
We show that ${S}_n(f)/l_n(f)$ is close to $\lambda_n$ in probability.
Let $D_n$ be the set of leaves that transmit directly to the fusion center.
The proof of the following lemma is in the Appendix.

\begin{Lemma}\label{lemma:SnConvergence}
For all $\eta > 0$, $\P_0( |{S}_n(f)/l_n(f) - \lambda_n| > \eta) \to 0$, as $n \to \infty$.
\end{Lemma}

We return to the proof of Proposition \ref{prop:NecessaryConditions}.
Given the transmission functions at all other nodes, the fusion center will optimize performance by using an appropriate likelihood ratio test, with a (possibly randomized) threshold. We can therefore assume, without loss of generality that this is the case. We let  $\zeta_n$ be the threshold chosen, and note that it must satisfy
\begin{equation}\label{eq:aleq}
\P_0(S_n(f)/l_n(f) \leq \zeta_n)\geq 1-\alpha.
\end{equation}
From a change of measure argument (see Lemma 3.4.7 in \cite{DemZei:98}), we have for $\eta >0$,
\begin{align*}
& \ofrac{l_n(f)}\log \beta^*(T_n) \\
& \geq \lambda_n - \eta
+ \ofrac{l_n(f)}\log\P_{0}\Big( \lambda_n - \eta < \frac{{S}_n(f)}{l_n(f)} \leq \zeta_n \Big).
\end{align*}
Using (\ref{eq:aleq}) and Lemma \ref{lemma:SnConvergence},
we see that
the last term goes to 0 as $n\to\infty$.We also have
\begin{align*}
\lambda_n
& = \ofrac{l_n(f)} \Big( \sum_{v \in D_n} \E_0 \big[\log \ddfrac{\P_1^{\gamma_v}}{\P_0^{\gamma_v}}\big]
+ \sum_{v \in F_{N,n}} \E_0[L_{v,n}] \Big)\\
& \geq (1-q_{N,n}) g^*_P + q_{N,n} K,
\end{align*}
where, using the notation in Assumption \ref{assumpt:LossOfDivergence},
\begin{align*}
K = \inf_{\substack{1<k\leq N \\ \xi\in\Gamma(k)\times\Gamma^k}} \ofrac{k} \E_0\Big[\log \ddfrac{\nu_1^\xi}{\nu_0^\xi}\Big]
> g^*_P.
\end{align*}
Then, letting $n\to\infty$, we have
\begin{align*}
g^*_R
& \geq (1-q_N) g^*_P + q_N K -\eta,
\end{align*}
for all $\eta > 0$. Taking $\eta \to 0$ completes the proof.
\end{proof}

The condition that there exists a finite $N$ such that
$l_n(v) \leq N$ for a non-vanishing proportion of nodes,
in the statement of Proposition \ref{prop:NecessaryConditions},
can be thought of as
corresponding to a situation where relay nodes are of two different types: high cost relays that can process a large number of received messages ($l_n(v) \to \infty$) and low cost relays that can only process a limited number of received messages
($l_n(v) \leq N$ for some small $N$). From this perspective,
Proposition \ref{prop:NecessaryConditions} states that a tree network of height greater than one, with a nontrivial proportion of low cost relays, will always have a performance worse than that of a parallel configuration.

Together with Proposition \ref{prop:NeymanPearsonSuff}, we have shown the following.
\begin{Proposition}\label{prop:NecessaryAndSufficient}
Suppose that Assumptions \ref{assumpt:EquivalentMeasures}-\ref{assumpt:LossOfDivergence} hold.
Then, $g^*_R=g^*_P$
iff $z=1$ (or equivalently, iff $q_N=0$ for all $N>0$).
\end{Proposition}

We close with an example in which $z<1$ and $g^*<g^*_P$. Since there are also easy examples where $z<1$ and $g^*_P<g^*$, this suggests that one can combine them to construct examples where $z<1$ and $g^*=g^*_P$. Thus, unlike the case of a relay tree, $z=1$ is not a necessary condition for $g^*=g_P$.

\begin{Example}\label{ex:gstarCounterEx}
Consider the tree network shown in Figure \ref{fig:gstarCounterEx}, where
every node makes a 3-bit observation. Each leaf then compresses its 3-bit observation
to a 1-bit message, while each level 1 node is allowed to send a 4-bit message.
(Recall that our framework allows for different transmission function sets $\Gamma(d)$ at the different levels.)
We assume Assumptions \ref{assumpt:EquivalentMeasures}-\ref{assumpt:LLRQ} hold.
Moreover, we assume that this network satisfies Assumption \ref{assumpt:LossOfDivergence}.

\begin{figure}[!htb]
\begin{center}
\includegraphics[scale=1]{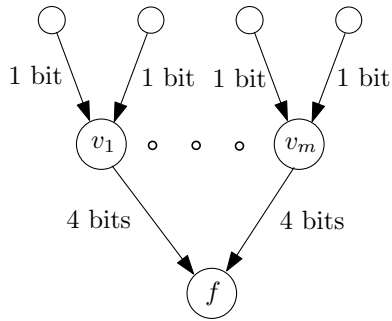}
\caption{Every node makes a 3-bit observation. Leaves are constrained to sending 1-bit messages,
while level 1 nodes are constrained to sending 4-bit messages.}\label{fig:gstarCounterEx}
\end{center}
\end{figure}

Consider the following strategy: each level 1 node forwards the two 1-bit messages it receives
from its two leaves to the fusion center. It then compress its own 3-bit observation into a 2-bit message before
sending it to the fusion center. Using this strategy, the tree network is equivalent to a parallel
configuration with $3m$ nodes, $2m$ of which are constrained to sending 1-bit messages, and $m$ of which
are constrained to sending 2-bit messages. Clearly, this parallel configuration performs strictly better
than one in which all $3m$ nodes are constrained to sending 1-bit messages, therefore we have $g^* < g_P^*$.
\hfill \EndExample
\end{Example}

Example \ref{ex:gstarCounterEx} shows that, unlike the case of relay trees, a tree can outperform
a parallel configuration.
On the other hand, Example \ref{ex:gstarCounterEx} is an artifact of our assumptions.
For example,
if we restrict every node in this example to sending only 1 bit,
the situation is reversed and we have $g^*_P < g^*$.
The question of
whether a parallel configuration always performs at least as well as
a tree network, i.e., whether $g_P^*\leq g^*$, when every node can send the same number of bits,
remains open.

\section{Conclusion}\label{sect:Conclusion}

We have studied the asymptotic detection performance of tree networks with bounded height, under a Neyman-Pearson criterion.
Similar to the parallel configuration, we have shown that the optimal Type II error probability decays exponentially fast with the number of
nodes. In addition, we have shown
that if leaves dominate (i.e., $l_n(f)/n \to 1$),
the network can achieve the same performance as if all nodes were transmitting directly to the fusion center.
We also
provided a simple strategy, in which
all leaves use the same transmission function, and all other nodes act as 1-bit relays, which
achieves the optimal error exponent to any desired accuracy.
The sufficient conditions are easy to achieve in cases of practical interest, hence a system
designer can obtain the optimal performance while ensuring that the network is energy efficient.
Once the sufficient conditions are satisfied, the architecture of the network no longer affects
its detection error exponent. On the other hand, we also showed that for the practically interesting case where $z=1$,
the sufficient conditions are also necessary.
Thus, in a network where the leaves do not dominate,
the error decay rate will be worse than that of a parallel configuration, and will actually
depend on the particular network architecture.

Needless to say, our conclusions only hold for the particular setting and criterion we have employed.
One issue that has not been touched upon is that, with a relay network, a significantly larger value of $n$ may be required before the asymptotic error exponent yields a good approximation.
Moreover, in practice, it
would be wasteful to have only the leaves make observations, if $n$ is not large enough.
Furthermore, under a Bayesian criterion, the
same performance as the parallel configuration can no longer be achieved,
although exponential decay is still possible \cite{TayTsiWin:C07b}.
Finally, the more realistic case where the i.i.d.\ assumption is violated, remains unexplored,
with work mainly limited to the parallel configuration \cite{DraLee:91,KamZhuGra:92,BluKas:92,BluKasPoo:97,ChaVee:06,LiDai:05}.

Future work includes characterizing the asymptotically optimal performance of tree networks without the bounded height constraint.
We would like to understand the rate at which the error probability decays, and its dependence on the rate at which the height of the tree increases.
Another intriguing question, which has been left unanswered, is whether the inequality $g^*_P\leq g^*$ is always true under the bounded height assumption, when every node is constrained to sending the same number of bits.

\section{Acknowledgements}\label{sect:Acknowledgements}

We wish to thank the anonymous reviewers for their
careful reading of the manuscript, and their detailed comments that have improved the
presentation. 

\appendix
\section{Appendix}

\subsection{Proof of Proposition \ref{prop:UpperBddErr}}

We first show part (\ref{it:general}).
The proof proceeds by induction on $k$. Suppose that $k=1$, which
is equivalent to the
well-studied case where all sensors transmit directly to a fusion center.
In this case, $p_n(v)=l_n(v)$.
Since $t_1 \in (-\KLD{\P_0^\gamma}{\P_1^\gamma}, \KLD{\P_1^\gamma}{\P_0^\gamma})$,
from (2.2.13) of \cite{DemZei:98}, we obtain
\begin{align*}
\ofrac{l_n(v)}\log \P_1\Big ( \frac{S_n(v)}{l_n(v)} \leq t_1 \Big )
& \leq -\Lambda_{1,1}^*(\gamma,t_1).
\end{align*}
The inequality for the Type I error probability follows from a similar argument.

Consider now the induction hypothesis
that the result holds for some $k$. Given a $k$-uniform tree rooted at $v$, the induction hypothesis leads to bounds on the probabilities associated with the
log-likelihood ratio $L_{v,n}$ of the message $Y_{v}$ computed at the node $v$. We use these bounds to obtain bounds on the log-moment generating function of $L_{v,n}$.
Recall that $L_{v,n}$ equals ${\cal L}_{v,n}(0)$ whenever $Y_{v}=0$, which is the case if and only if $S_n(v)/l_n(v)\leq t_k$.
Fix some $\lambda\in[-1,0]$. We have
\begin{align}
& \ofrac{l_n(v)} \log \E_1 \big[e^{\lambda L_{v,n}}\big] \nonumber\\
& = \ofrac{l_n(v)} \log \Big[ \P_1(Y_v=0) e^{\lambda \mathcal{L}_{v,n}(0)}
+ \P_1(Y_v=1) e^{\lambda \mathcal{L}_{v,n}(1)} \Big] \nonumber\\
& = \ofrac{l_n(v)} \log \Big[
\P_1(Y_v=0)^{1+\lambda}\P_0(Y_v=0)^{-\lambda}
+
\P_1(Y_v=1)^{1+\lambda}\P_0(Y_v=1)^{-\lambda}
\Big] \nonumber\\
& \leq \ofrac{l_n(v)} \log \Big[
\P_1(Y_v=0)^{1+\lambda}
+ \P_0(Y_v=1)^{-\lambda}
\Big].\nonumber
\end{align}
Using the inequality $\log(a+b)\leq \max\{\log(2a),\log(2b)\}$, we obtain
\begin{align}
& \ofrac{l_n(v)} \log \E_1\big[ e^{\lambda L_{v,n}}\big] \nonumber\\
& \leq \max \big\{\frac{1+\lambda}{l_n(v)} \log\P_1(Y_v=0), - \frac{\lambda}{l_n(v)} \log\P_0(Y_v=1) \big\}
+ \frac{\log 2}{l_n(v)} \nonumber\\
& \leq \max \big\{ -(1+\lambda) \Lambdahs{1}{k}, \lambda \Lambdahs{0}{k} \big\}
+ \frac{p_n(v)}{l_n(v)}-1 + \frac{\log 2}{l_n(v)} \label{InductionStep}\\
& \leq \Lambda_{1,k}(\gamma,t\tn{k};\lambda) + \frac{p_n(v)}{l_n(v)}+ \frac{1}{l_n(v)}-1,\label{LogCumh}
\end{align}
where (\ref{InductionStep}) follows from the induction hypothesis.

Consider now a node $u$ at level $k+1$. The subtree rooted at $u$ is a $(k+1)$-uniform tree. Each level $k$ node $v\in C_n(u)$ can be viewed as the root of a $k$-uniform tree and Eq.\ (\ref{LogCumh}) can be applied to $L_{v,n}$.
From the Markov Inequality, and since $\lambda \in [-1,0]$, we have
\begin{align*}
\P_1 \Big( \frac{S_n(u)}{l_n(u)} \leq t_{k+1} \Big) & \leq e^{-\lambda l_n(u) t_{k+1}} \E_1 \big[e^{\lambda S_n(u)}\big],
\end{align*}
so that
\begin{align}
& \ofrac{l_n(u)} \log\P_1 \Big( \frac{S_n(u)}{l_n(u)} \leq t_{k+1} \Big) \nonumber\\
& \leq - \lambda t_{k+1} + \ofrac{l_n(u)} \sum_{v \in C_n(u)} \log \E_1 \big[e^{\lambda L_{v,n}}\big] \nonumber\\
& = - \lambda t_{k+1}
+ \sum_{v \in C_n(u)} \frac{l_n(v)}{l_n(u)} \cdot \ofrac{l_n(v)} \log \E_1 \big[e^{\lambda L_{v,n}} \big] \nonumber\\
& \leq - \lambda t_{k+1} + \Lambda_{1,k}(\gamma,t\tn{k};\lambda)
+ \sum_{v \in C_n(u)}\frac{p_n(v)}{l_n(u)}+ \frac{|C_n(u)|}{l_n(u)}-1 \label{cumulantbdd1}\\
& = - \lambda t_{k+1} + \Lambda_{1,k}(\gamma,t\tn{k};\lambda) + \frac{p_n(u)}{l_n(u)} - 1,\label{cumulantbdd}
\end{align}
where (\ref{cumulantbdd1}) follows from the induction hypothesis and (\ref{LogCumh}).
Taking the infimum over $\lambda \in [-1,0]$ (cf.\ Lemma \ref{lemma:LambdaRelationship}), and using (\ref{eq:lstar}),
we obtain
\begin{align*}
\ofrac{l_n(u)} \log\P_1 \Big( \frac{S_n(u)}{l_n(u)} \leq t_{k+1} \Big)
\leq - \Lambdahs{1}{k+1} + \frac{p_n(u)}{l_n(u)} - 1.
\end{align*}
A similar argument proves the result for the Type I error probability,
and the proof of part (\ref{it:general}) is complete.

For part (\ref{it:N}), suppose that for all $n\geq n_0$ and all $v \in {B}_n$, we have $l_n(v) \geq N$. Note that
$l_n(f) \geq N|B_n|$. Furthermore, the number of nodes at each level $k \geq 1$ is bounded by $|B_n|$, which yields
$$
\frac{p_n(f)}{l_n(f)} -1\leq
\frac{n}{l_n(f)} - 1
= \frac{n - l_n(f)}{l_n(f)} \\
\leq \frac{h |B_n|}{N |B_n|} =\frac{h}{N}.
$$
Applying the results from part (\ref{it:general}), with $k=h$, we obtain  part
(\ref{it:N}).


\subsection{Proof of Lemma \ref{lemma:Uniformizedq_N}}

We have
$l_n'(f)=l_n(f)$. Furthermore, it can be shown that $|B_n'| \leq h |B_n|$. Therefore,
\begin{align}
q'_{N, n} = \ofrac{l_n'(f)}\sum_{v\in F'_{N,n}} l_n'(v)
& \leq \ofrac{l_n(f)} N |B_n'| \nonumber\\
& \leq \ofrac{l_n(f)} N h \big( |F_{M,n}| + |F_{M,n}^c|\big)\nonumber\\
& \leq h N q_{M,n} + h N/M,\nonumber
\end{align}
where the last inequality follows from
$|F_{M,n}| \leq \sum\limits_{v \in F_{M,n}} l_n(v)$ and $|F_{M,n}^c| \leq l_n(f)/M$.
Taking the limit superior as $n\to\infty$, we obtain
\begin{align*}
q'_N \leq h (N q_M + N/M).
\end{align*}
Suppose that $q_M=0$ for all $M>0$. Then for all $N,M>0$, we have
\begin{align*}
q'_N \leq h N/M.
\end{align*}
Taking $M\to \infty$, we obtain the desired result.


\subsection{Proof of Lemma \ref{lemma:Equivalence}}

Suppose that $q_N > 0$ for some $N > 0$.
Using the inequality
\begin{align*}
& q_{N,n} = \ofrac{l_n(f)} \sum_{v \in F_{N,n}} l_n(v) \leq \frac{N|F_{N,n}|}{l_n(f)},\end{align*}
or
\begin{equation}\label{eq:F}
|F_{N,n}| \geq \frac{q_{N,n}}{N} l_n(f),
\end{equation}
we obtain
\begin{align*}
\frac{l_n(f)}{n}
& \leq \frac{l_n(f)}{|F_{N,n}| + l_n(f)}\\
& \leq  \frac{l_n(f)}{q_{N,n} l_n(f)/N + l_n(f)} \\
& = \frac{N}{N + q_{N,n}}.
\end{align*}
Letting $n \to \infty$, we obtain
\begin{align*}
z \leq \frac{N}{N + q_{N}} < 1.
\end{align*}

For the converse, suppose that $q_N=0$ for all $N>0$.
It can be seen that each non-leaf node is on a path that connects some $v\in B_n$ to the fusion center.
Therefore, the number of non-leaf nodes $n-l_n(f)$ is bounded by $h|B_n|$.
We have
$$
\frac{n - l_n(f)}{l_n(f)}
\leq \frac{h|B_n|}{l_n(f)}
= h\frac{|F_{N,n}| + |F_{N,n}^c|}{l_n(f)}
 \leq h q_{N,n} + \frac{h}{N}.
$$
Therefore,
\begin{align*}
\limsup_{n\to\infty} \frac{n - l_n(f)}{l_n(f)}
& \leq \frac{h}{N}.
\end{align*}
This is true for all $N >0$, which implies that $\lim\limits_{n\to\infty} l_n(f)/n=1$.


\subsection{Proof of Lemma \ref{lemma:SnConvergence}}

For each $v\in B_n$, we have $Y_{v} = \gamma_v( \{\gamma_u(X_u): u \in C_n(v)\})$, for some $\gamma_v \in \Gamma(l_n(v))$.
Using the first, and the second part of Lemma \ref{lemma:BoundedDivergence}, there exists some $a_1 \in (0,\infty)$, such that
\begin{align}
\E_0[L_{v,n}^2]
& \leq \E_0 \Big[ \Big(\sum_{u\in C_n(v)} \log \ddfrac{\P_1^{\gamma_u}}{\P_0^{\gamma_u}} \Big)^2 \Big] + 1 \nonumber\\
& \leq l_n(v)\E_0 \Big[ \sum_{u\in C_n(v)} \log^2 \ddfrac{\P_1^{\gamma_u}}{\P_0^{\gamma_u}}  \Big] + 1   \nonumber\\
& \leq l^2_n(v) a_1 + 1 \nonumber\\
& \leq l_n^2(v)a, \label{log2Bdd}
\end{align}
where $a = a_1+1$.

To prove the lemma, we use Chebychev's inequality, and the inequalities $l_n(v)\leq N$ for $v\in F_{N,n}$, and $|D_n|\leq l_n(f)$, to obtain
\begin{align}
& \P_0\Big( \big| \frac{{S}_n(f)}{l_n(f)} - \lambda_n \big| > \eta \Big) \nonumber\\
& \leq \ofrac{\eta^2 l_n^2(f)} \Big( \sum_{v\in D_n}
\E_0\big[\log^2 \ddfrac{\P_1^{\gamma_v}}{\P_0^{\gamma_v}}\big]
+ \sum_{v\in F_{N,n}} \E_0[L_{v,n}^2] \Big) \nonumber\\
& \leq \ofrac{\eta^2 l_n^2(f)} \Big( \sum_{v\in D_n} a
+ \sum_{v\in F_{N,n}} l_n^2(v) a \Big) \label{BoundedLog2}\\
& \leq \frac{a}{\eta^2 l_n(f)} + \frac{a}{\eta^2 l_n(f)} \sum_{v\in F_{N,n}} \frac{l_n(v)}{l_n(f)} N \nonumber\\
& \leq \frac{a(1+N)}{\eta^2 l_n(f)}, \label{ChebyBdd}
\end{align}
where (\ref{BoundedLog2}) follows from Lemma \ref{lemma:BoundedDivergence} and (\ref{log2Bdd}). The R.H.S.\ of (\ref{ChebyBdd}) goes to zero as $n\to \infty$, and the proof is complete.



\begin{thebibliography}{10}
\providecommand{\url}[1]{#1}
\csname url@samestyle\endcsname
\providecommand{\newblock}{\relax}
\providecommand{\bibinfo}[2]{#2}
\providecommand{\BIBentrySTDinterwordspacing}{\spaceskip=0pt\relax}
\providecommand{\BIBentryALTinterwordstretchfactor}{4}
\providecommand{\BIBentryALTinterwordspacing}{\spaceskip=\fontdimen2\font plus
\BIBentryALTinterwordstretchfactor\fontdimen3\font minus
  \fontdimen4\font\relax}
\providecommand{\BIBforeignlanguage}[2]{{%
\expandafter\ifx\csname l@#1\endcsname\relax
\typeout{** WARNING: IEEEtran.bst: No hyphenation pattern has been}%
\typeout{** loaded for the language `#1'. Using the pattern for}%
\typeout{** the default language instead.}%
\else
\language=\csname l@#1\endcsname
\fi
#2}}
\providecommand{\BIBdecl}{\relax}
\BIBdecl

\bibitem{TenSan:81}
R.~R. Tenney and N.~R. Sandell, ``Detection with distributed sensors,''
  \emph{{IEEE} Trans. Aerosp. Electron. Syst.}, vol.~17, pp. 501--510, 1981.

\bibitem{ChaVar:86}
Z.~Chair and P.~K. Varshney, ``Optimal data fusion in multiple sensor detection
  systems,'' \emph{{IEEE} Trans. Aerosp. Electron. Syst.}, vol.~22, pp.
  98--101, 1986.

\bibitem{PolTsi:90}
G.~Polychronopoulos and J.~N. Tsitsiklis, ``Explicit solutions for some simple
  decentralized detection problems,'' \emph{{IEEE} Trans. Aerosp. Electron.
  Syst.}, vol.~26, pp. 282--292, 1990.

\bibitem{WilWar:92}
P.~Willett and D.~Warren, ``The suboptimality of randomized tests in
  distributed and quantized detection systems,'' \emph{{IEEE} Trans. Inf.
  Theory}, vol.~38, pp. 355--361, Mar. 1992.

\bibitem{Tsi:93a}
J.~N. Tsitsiklis, ``Extremal properties of likelihood-ratio quantizers,''
  \emph{{IEEE} Trans. Commun.}, vol.~41, pp. 550--558, 1993.

\bibitem{Tsi:93}
------, ``Decentralized detection,'' \emph{Advances in Statistical Signal
  Processing}, vol.~2, pp. 297--344, 1993.

\bibitem{IrvTsi:94}
W.~W. Irving and J.~N. Tsitsiklis, ``Some properties of optimal thresholds in
  decentralized detection,'' \emph{{IEEE} Trans. Autom. Control}, vol.~39, pp.
  835--838, 1994.

\bibitem{VisVar:97}
R.~Viswanathan and P.~K. Varshney, ``Distributed detection with multiple
  sensors: part {I} - fundamentals,'' \emph{Proc. {IEEE}}, vol.~85, pp. 54--63,
  1997.

\bibitem{CheVar:02}
B.~Chen and P.~K. Varshney, ``A {Bayesian} sampling approach to decision fusion
  using hierarchical models,'' \emph{{IEEE} Trans. Signal Process.}, vol.~50,
  no.~8, pp. 1809--1818, Aug. 2002.

\bibitem{CheWil:05}
B.~Chen and P.~K. Willett, ``On the optimality of the likelihood-ratio test for
  local sensor decision rules in the presence of nonideal channels,''
  \emph{{IEEE} Trans. Inf. Theory}, vol.~51, no.~2, pp. 693--699, Feb. 2005.

\bibitem{Kas:06}
A.~Kashyap, ``Comments on ``{On} the optimality of the likelihood-ratio test
  for local sensor decision rules in the presence of nonideal channels'',''
  \emph{{IEEE} Trans. Inf. Theory}, vol.~52, no.~3, pp. 1274--1275, Mar. 2006.

\bibitem{LiuChe:06}
B.~Liu and B.~Chen, ``Channel-optimized quantizers for decentralized detection
  in sensor networks,'' \emph{{IEEE} Trans. Inf. Theory}, vol.~52, no.~7, pp.
  3349--3358, Jul. 2006.

\bibitem{EkcTen:82}
L.~K. Ekchian and R.~R. Tenney, ``Detection networks,'' in \emph{Proc. {IEEE}
  {C}onference on {D}ecision and {C}ontrol}, 1982, pp. 686--691.

\bibitem{VisThoTum:88}
R.~Viswanathan, S.~C.~A. Thomopoulos, and R.~Tumuluri, ``Optimal serial
  distributed decision fusion,'' \emph{{IEEE} Trans. Aerosp. Electron. Syst.},
  vol.~24, no.~4, pp. 366--376, 1988.

\bibitem{ReiNol:90a}
A.~R. Reibman and L.~W. Nolte, ``Design and performance comparison of
  distributed detection networks,'' \emph{{IEEE} Trans. Aerosp. Electron.
  Syst.}, vol.~23, pp. 789--797, 1987.

\bibitem{TanPatKle:91}
Z.~B. Tang, K.~R. Pattipati, and D.~L. Kleinman, ``Optimization of detection
  networks: part {I} - tandem structures,'' \emph{{IEEE} Trans. Syst., Man,
  Cybern.}, vol.~21, no.~5, pp. 1044--1059, 1991.

\bibitem{TanPatKle:93}
------, ``Optimization of detection networks: part {II} - tree structures,''
  \emph{{IEEE} Trans. Syst., Man, Cybern.}, vol.~23, no.~1, pp. 211--221, 1993.

\bibitem{PapAth:92a}
J.~D. Papastavrou and M.~Athans, ``On optimal distributed decision
  architectures in a hypothesis testing environment,'' \emph{{IEEE} Trans.
  Autom. Control}, vol.~37, no.~8, pp. 1154--1169, 1992.

\bibitem{PetPatKle:94}
A.~Pete, K.~Pattipati, and D.~Kleinman, ``Optimization of detection networks
  with multiple event structures,'' \emph{{IEEE} Trans. Autom. Control},
  vol.~39, no.~8, pp. 1702--1707, 1994.

\bibitem{AlhVar:95}
S.~Alhakeem and P.~K. Varshney, ``A unified approach to the design of
  decentralized detection systems,'' \emph{{IEEE} Trans. Aerosp. Electron.
  Syst.}, vol.~31, no.~1, pp. 9--20, 1995.

\bibitem{LinCheVar:05}
Y.~Lin, B.~Chen, and P.~K. Varshney, ``Decision fusion rules in multi-hop
  wireless sensor networks,'' \emph{{IEEE} Trans. Aerosp. Electron. Syst.},
  vol.~41, no.~2, pp. 475--488, Apr. 2005.

\bibitem{Tsi:88}
J.~N. Tsitsiklis, ``Decentralized detection by a large number of sensors,''
  \emph{Math. Control, Signals, Syst.}, vol.~1, pp. 167--182, 1988.

\bibitem{HelCov:70}
M.~E. Hellman and T.~M. Cover, ``Learning with finite memory,'' \emph{Ann. of
  Math. Statist.}, vol.~41, no.~3, pp. 765--782, 1970.

\bibitem{Cov:69}
T.~M. Cover, ``Hypothesis testing with finite statistics,'' \emph{Ann. of Math.
  Statist.}, vol.~40, no.~3, pp. 828--835, 1969.

\bibitem{PapAth:92}
J.~D. Papastavrou and M.~Athans, ``Distributed detection by a large team of
  sensors in tandem,'' \emph{{IEEE} Trans. Aerosp. Electron. Syst.}, vol.~28,
  no.~3, pp. 639--653, 1992.

\bibitem{TayTsiWin:C07c}
W.~P. Tay, J.~N. Tsitsiklis, and M.~Z. Win, ``On the sub-exponential decay of
  detection probabilities in long tandems,'' in \emph{Proc. IEEE Int. Conf.
  Acoustics, Speech, and Signal Processing}, Honolulu, HI, Apr. 2007, pp. 837
  -- 840.

\bibitem{DemZei:98}
A.~Dembo and O.~Zeitouni, \emph{Large Deviations Techniques and
  Applications}.\hskip 1em plus 0.5em minus 0.4em\relax New York, NY:
  Springer-Verlag, 1998.

\bibitem{TayTsiWin:C07b}
W.~P. Tay, J.~N. Tsitsiklis, and M.~Z. Win, ``Bayesian detection in bounded
  height tree networks,'' in \emph{Proc. {D}ata {C}ompression {C}onf.},
  Snowbird, UT, Mar. 2007, pp. 243 -- 252.

\bibitem{DraLee:91}
E.~Drakopoulos and C.~C. Lee, ``Optimum multisensor fusion of correlated local
  decisions,'' \emph{{IEEE} Trans. Aerosp. Electron. Syst.}, vol.~27, no.~4,
  pp. 593--606, Jul. 1991.

\bibitem{KamZhuGra:92}
M.~Kam, Q.~Zhu, and W.~S. Gray, ``Optimal data fusion of correlated local
  decisions in multiple sensor detection systems,'' \emph{{IEEE} Trans. Aerosp.
  Electron. Syst.}, vol.~28, no.~3, pp. 916--920, 1992.

\bibitem{BluKas:92}
R.~S. Blum and S.~A. Kassam, ``Optimum distributed detection of weak signals in
  dependent sensors,'' \emph{{IEEE} Trans. Inf. Theory}, vol.~38, no.~3, pp.
  1066--1079, May 1992.

\bibitem{BluKasPoo:97}
R.~S. Blum, S.~A. Kassam, and H.~Poor, ``Distributed detection with multiple
  sensors: {part II} - advanced topics,'' \emph{Proc. {IEEE}}, vol.~85, no.~1,
  pp. 64--79, 1997.

\bibitem{ChaVee:06}
J.-F. Chamberland and V.~V. Veeravalli, ``How dense should a sensor network be
  for detection with correlated observations?'' \emph{{IEEE} Trans. Inf.
  Theory}, vol.~52, no.~11, pp. 5099--5106, Nov. 2006.

\bibitem{LiDai:05}
W.~Li and H.~Dai, ``Distributed detection in large-scale sensor networks with
  correlated sensor observations,'' in \emph{Proc. Allerton Conf. on
  Communication, Control, and Computing}, Sep. 2005.

\end{thebibliography}


\end{document}